\documentclass[10pt,a4paper,dvipsnames,usenames]{article}
\usepackage{amssymb}
\usepackage{amsmath}
\usepackage{amsthm}
\usepackage{latexsym}
\usepackage[dvips]{epsfig}
\usepackage{mathrsfs}
\usepackage{eufrak}
\usepackage{bm}
\usepackage{stmaryrd}
\usepackage{authblk}

\usepackage[dvipsnames]{xcolor}
\usepackage{tikz}
\usepackage{enumitem}
\usepackage{slashed}
\usepackage{bbm}

\theoremstyle{plain}
\newtheorem{proposition}{Proposition}
\newtheorem{lemma}{Lemma}

\newtheorem{assumption}{Assumption}

\allowdisplaybreaks

\def\bma{{\bm a}}
\def\bmb{{\bm b}}
\def\bmc{{\bm c}}
\def\bmd{{\bm d}}
\def\bme{{\bm e}}
\def\bmf{{\bm f}}
\def\bmg{{\bm g}}

\def\bmi{{\bm i}}
\def\bmj{{\bm j}}
\def\bmk{{\bm k}}
\def\bml{{\bm l}}

\def\bmm{{\bm m}}

\def\bmzero{{\bm 0}}
\def\bmone{{\bm 1}}

\def\bmthree{{\bm 3}}

\def\bmA{{\bm A}}
\def\bmB{{\bm B}}
\def\bmC{{\bm C}}
\def\bmD{{\bm D}}
\def\bmE{{\bm E}}
\def\bmF{{\bm F}}

\def\bmP{{\bm P}}
\def\bmQ{{\bm Q}}

\def\bmphi{{\bm \phi}}

\def\bmupsilon{{\bm \upsilon}}

\def\bmGamma{{\bm \Gamma}}

\def\bmUpsilon{{\bm \Upsilon}}
\def\bmSigma{{\bm \Sigma}}
\def\bmLambda{{\bm \Lambda}}

\newcommand{\Ss}{\mathcal{S}}
\newcommand{\Sstar}{\mathcal{S}_\star}

\newcommand{\T}{\mathcal{T}}

\newcommand{\TD}{D}

\usepackage{mathtools}
\usepackage{xparse}
\makeatletter
\newcommand{\raisemath}[1]{\mathpalette{\raisem@th{#1}}}
\newcommand{\raisem@th}[3]{\raisebox{#1}{$#2#3$}}
\makeatother
\NewDocumentCommand{\newrbar}{O{0pt} O{0pt}}{
  \ensuremath{\mathrlap{\raisemath{#2}{\hspace*{#1}{\mathchar'26\mkern-9mu}}}r}}

\newcounter{mnotecount}

\newcommand{\mnotex}[1]
{\protect{\stepcounter{mnotecount}}$^{\mbox{\footnotesize $\bullet$\themnotecount}}$ 
\marginpar{
\raggedright\tiny\em
$\!\!\!\!\!\!\,\bullet$\themnotecount: #1} }

\newcounter{mnote}

\usepackage{color}
\usepackage{xcolor}
\usepackage[normalem]{ulem}
\usepackage{soul}
\usepackage{hyperref}

\usepackage{cleveref}

\renewcommand{\leq}{\leqslant}

\begin{document}

\title{\textbf{Wellposedness of the initial boundary value problem for the conformal field equations}}

\author[1]{Chris Stevens \footnote{E-mail address:{\tt chris.stevens@canterbury.ac.nz}}}
\author[2]{Juan A. Valiente Kroon \footnote{E-mail address:{\tt j.a.valiente-kroon@qmul.ac.uk}}}

\affil[1]{School of Mathematics and Statistics,
  University of Canterbury, Christchurch 8041, New Zealand}

\affil[2]{School of Mathematical Sciences, Queen Mary, University of London,
Mile End Road, London E1 4NS, United Kingdom}

\maketitle

\begin{abstract}
We provide a formulation of the initial boundary value problem for Friedrich's extended conformal Einstein field equations in which boundary data is prescribed on a timelike hypersurface located at a finite position in the spacetime. Our construction relies on a gauge based on the properties of conformal geodesics and requires the the boundary is ruled by timelike conformal geodesics. The consequences of this assumption on the timelike boundary are analysed and we identify a subset of maximally dissipative boundary conditions which are consistent with this assumption. For this class of consistent boundary conditions we establish the wellposedness of the initial boundary value problem and prove the propagation of the constraints. 
\end{abstract}



\section{Introduction}
The conformal Einstein field equations, originally introduced by
Helmut Friedrich, provide a powerful framework for studying the global
structure of solutions to the Einstein equations by extending the
spacetime to include its conformal boundary. In this setting,
asymptotic properties of gravitational fields such as the outgoing
waveform and conserved quantities can be analysed in a regular and
geometrically transparent manner.

In a series of works \cite{beyer2017numerical,frauendiener2021non,
frauendiener2023non, frauendiener2023non2, frauendiener2025fully} a
general class of Initial Boundary Value Problems (IBVPs) has been
considered in which the outer boundary is taken to be a timelike
hypersurface located at a finite position in the spacetime and ruled
by a congruence of timelike conformal geodesics. These constructions,
based on conformal Gaussian gauges, have been successfully implemented
in numerical simulations incorporating the conformal boundary within
the computational domain and exhibit stable and convergent behaviour,
with constraint violations remaining well-controlled. Nevertheless, a
rigorous proof of wellposedness for this general setting has so far
remained an open problem.

A proof of wellposedness does exist, but for the special case of
asymptotically Anti-de Sitter spacetimes
\cite{friedrich1995einstein}. In this context, the conformal boundary
is timelike and can be treated as the outer boundary of the spacetime
domain. Friedrich showed that, under these circumstances, one can
obtain a wellposed IBVP by exploiting the special geometric properties
of the conformal boundary together with maximally dissipative boundary
conditions. The key simplification in this setting stems from the fact
that the boundary is not only timelike but also geometrically
distinguished, allowing for a natural compatibility between the
conformal structure, the gauge and the boundary conditions.

Beyond the Anti-de Sitter case, however, the situation is considerably
more subtle. The purpose of this article is to close this gap. Our
approach is based on the development of a spinorial formalism adapted
to timelike boundaries, which can be viewed as a natural analogue of
the standard space-spinor formalism \cite{sommers1980space}. In
contrast to the latter, where the decomposition is performed with
respect to a timelike vector field, our construction employs a
spacelike vector field transverse to the boundary. This has the
consequence that the spinors intrinsic to the boundary, called
$SU(1,1)$\emph{-spinors}, can be further decomposed into temporal and
spatial pieces. This leads to a 1+1+2 decomposition tailored to the
geometry of timelike hypersurfaces, allowing for a refined analysis of
both the evolution system and the constraints intrinsic to the
boundary.

Using this formalism, we establish the wellposedness of the IBVP for
the extended conformal Einstein field equations in a neighbourhood of
the corner where the initial and boundary hypersurfaces meet. Our
result generalises Friedrich’s wellposedness result for Anti-de Sitter
spacetimes to a broader class of timelike boundaries. Specifically, we
assume that the boundary hypersurface is ruled by a congruence of
timelike conformal geodesics and satisfies a mild geometric condition
ensuring compatibility with the conformal Gau\ss ian gauge.

A key ingredient in our analysis is the use of a boundary-adapted
symmetric hyperbolic reduction of the Bianchi subsystem. This choice
of reduction ensures that the subsidiary system governing the
propagation of the constraints does not admit modes propagating
transversely to the boundary. As a consequence, the constraint
propagation reduces to a transport system intrinsic to the boundary
hypersurface. Furthermore, by analysing the conformal constraint
equations intrinsic to the boundary hypersurface we show that the
boundary system is fully determined once maximally dissipative (in
particular, fully reflective) boundary conditions are imposed together
with a single complex free datum. This result highlights the delicate
interplay between the geometry of the boundary, the conformal
constraints, and the admissible boundary data. Taken together, these
results provide a rigorous mathematical foundation for the class of
initial boundary value problems previously employed in numerical
studies.


\subsection{Notation and conventions} The signature convention for
4-dimensional Lorentzian metrics is $(+,-,-,-)$. Objects in a
\emph{physical spacetime} satisfying the vacuum Einstein equations
carry a $\tilde{\phantom{X}}$ as ornament. Objects in the conformally
rescaled \emph{unphysical spacetime} carry no ornament. We make use of
the \emph{abstract} index conventions as described in
\cite{penrose1984spinors} ---see also \cite{kroon2017conformal}. In
particular, the lowercase letters $a,\, b,\, c,\ldots$ are used as
abstract spacetime indices, while $i,\,j\,k,\ldots$ are used as
indices for tensors in 3-dimensional manifolds ---either Riemannian or
Lorentzian. Boldface indices $\bma,\, \bmb,\,\bmc,\ldots$ are used as
spacetime frame indices (i.e. the components of tensors with respect
to a tetrad) while $\bmi,\,\bmj,\,\bmk,\ldots$ are used to denote the
components with respect to a triad. We use Greek letters $\mu,\, \nu,
\, \lambda,\ldots $ as spacetime coordinate indices while $\alpha,\,
\beta,\, \gamma,\ldots $ are used as coordinate indices in a
3-dimensional submanifold. In addition, we occasionally employ
boldface index-free notation for geometric objects. In this notation
tensors are written without explicit indices and are understood as
geometric operators or multilinear maps; for example, the spacetime
metric $g_{ab}$ may simply be written as $\bmg$.

Spinors are used systematically in in various parts of the article. We
follow the conventions of \cite{penrose1984spinors}. In particular, we
use the uppercase Latin letters $A,\, B,\, C,\,\ldots$ as abstract
spinor indices. As in the case of tensor indices, boldface indices
like $\bmA,\, \bmB,\, \bmC,\,\ldots$ are used to denote components of
a spinor with respect to a given basis.

\subsection{Structure of the article}\label{Section:Structure} The
article is organised as follows: Section 2 provides an overview of the
extended conformal Einstein field equations and their hyperbolic
reduction in a conformal Gau\ss ian gauge, including both tensorial
and spinorial formulations. Section 3 introduces the geometric setting
of the initial boundary value problem, together with a gauge and frame
adapted to a timelike boundary, and analyses the restrictions implied
by this construction. Section 4 is devoted to the conformal constraint
equations intrinsic to the boundary hypersurface and the geometric
consequences of the assumptions imposed on the boundary. Section 5
formulates the initial boundary value problem, discusses admissible
boundary conditions, and establishes a local existence result. Section
6 analyses the propagation of the constraints. Finally, Section 7
studies the compatibility between the conformal Gau\ss ian gauge
adapted to the boundary and the maximally dissipative boundary
conditions.

\subsection{Outline of the argument}\label{Section-argumentoutline}
First, we put forward a short summary of the argument and proceed to
expand upon these in the body of the article.

\subsubsection*{Gauge} We make use of a gauge (F-gauge) based on the
properties of conformal geodesics. We require that the timelike
boundary $\T$ is ruled by conformal geodesics. Any conformal geodesic
that starts in $\T$ remains in $\T$. Moreover, there are no other
conformal geodesics entering the timelike boundary. This gauge sets
some restrictions on the properties of $\T$. A simple setting to
understand these restrictions is the so-called \emph{umbilical case}
---when the extrinsic curvature of $\T$ is pure trace.

\subsubsection*{The evolution system} We make use of the extended
conformal Einstein field equations (XCFE). We make use of a spinorial
version of these equations ---however, this is not essential. The
hyperbolic reductions for these equations in the F-gauge are well
understood and studied. The most notable feature is the way the
evolution equations split. Schematically, one can split the equation
into $\mathbf{\upsilon}$-variables encoding the components of the
frame $\bm{e}$, connection $\bm{\Gamma}$ and Schouten tensor $\bm{L}$
and $\mathbf{\phi}$-variables containing the independent components of
the rescaled Weyl tensor. The equations then take the form

\begin{eqnarray*}
&& \partial_\tau \bm{\upsilon} = \bm{K}\bm{\upsilon} + \bm{Q}(\bm{\Gamma})\bm{\upsilon} + \bm{L}(x)\bm{\phi}, \\
&& (\bm{I} + \bm{A}^0(\bm{e}))\partial_t\bm{\phi} + \bm{A}^\alpha(\bm{e})\partial_\alpha\bm{\phi}
= \bm{B}(\bm{\Gamma})\bm{\phi}.
\end{eqnarray*}

There are many ways to obtain a hyperbolic reduction of the Bianchi
subsystem. The F-gauge has the consequence that the
$\upsilon$-variables are governed by transport equations along the
conformal geodesics the gauge is adapted to, while the
$\phi$-variables yield a symmetric hyperbolic first order
system. Considering for a moment $\T$ as an outer boundary, a
characteristic analysis of this system will in general show that there
is more than one mode propagating inward transverse to $\T$. One mode
can be identified as the ingoing gravitational radiation, while the
others are gauge. In the present analysis it is essential to make use
of the \emph{boundary adapted system}, a specific choice of hyperbolic
reduction simplifying this characteristic structure to containing only
one ingoing mode. In this case the boundary data is simply a
manifestation of the single complex, physical degree of freedom
inherent in the Einstein equations and is unconstrained in this sense.

\subsubsection*{Boundary conditions}

The evolution system breifly described in the previous paragraph needs
to be supplemented with boundary and initial value conditions in such
a way that the theory of \emph{maximally disipative initial value
problems} can be applied. This class of boundary conditions is
obtained from analysing the derivatives in the evolution system along
the normal direction to $\T$. The transport subsystem does not contain
any normal derivatives. Thus, it does not contribute to the boundary
conditions. By the way the gauge is fixed, this subsystem is a
transport system on $\T$.

The Bianchi subsystem does contribute to the boundary conditions. The
theory of maximally dissipative systems fixes these as a certain
linear combination of the components of $\bm{\phi}$. In
Sec.~\ref{SubSection:MaximallyDissBCs} we show that the remaining
components of $\bm{\phi}$ can be computed from the maximally
dissipative prescriptions making used of the Bianchi--Gauss constaints
on $\T$. The above can be verified by means of a ``1+1+2'' spinor
formalism in which the components of the Weyl spinor prescribed
through the maximally dissipative boundary conditions act as sources
of the remaining unspecified components.

\subsubsection*{Corner conditions}
An important technical point that needs to be discussed (but not
solved!) are the compatibility conditions (corner conditions) between
the initial data on the initial hypersurface $\mathcal{S}_\star$ and
the boundary data on $\T$. This is a difficult problem and there is no
general procedure to to address this. \emph{Here, we will simply
assume that the corner conditions can be somehow implemented}. The
same concern is raised for the initial boundary value problem
formulation by Friedrich and Nagy \cite{friedrich1999initial}. They
mention that the corner conditions are satisfied so long as the
boundary conditions are chosen to be compatible with the formal
expansion of the initial data away from the initial surface.

\subsubsection*{Propagation of the constraints}
The propagation of the constraints in the boundary adapted gauge are
satisfied without any further work due to the subsidiary system
becoming a transport equation along $\T$. See for example Theorem 17.1
in \cite{kroon2017conformal}. This is a significant simplification
from the boundary adapted gauge choice.

\subsubsection*{Dropping the boundary adapted gauge restriction}
If one did not use the boundary adapted gauge then there would be no
restriction on the boundary data for $\bm{\phi}$ other than the
considerations coming from the maximally dissipative boundary
conditions. In general circumstances, this results in other ingoing
characteristic modes on $\T$ in both the evolution and subsidiary
systems. The new non-zero speeds for the evolution and subsidiary
systems turn out to be identical, a known result for the Bianchi
subsystem \cite{alcubierre2008introduction}. One can exploit this
relationship to fix the additional gauge boundary uniquely for the
evolution system by the requirement that the subsidiary system's
ingoing free data vanishes \cite{beyer2017numerical}. This is a
significantly more complicated case, and we do not consider it here.

\section{The extended conformal Einstein field
equations}\label{Section:CFE} This section provides a quick overview
of the extended conformal Einstein field equations. The motivation
behind is mostly that of settling the notation. For further details on
the derivation of these equations and how to obtain hyperbolic
reductions, the reader is referred to \cite{kroon2017conformal}.

\subsection{The geometric setting} In the following let
$(\tilde{\mathcal{M}},\tilde{\bmg})$, with $\tilde{\mathcal{M}}$ a
4-dimensional manifold and $\tilde{\bmg}$ a Lorentzian metric, denote
a vacuum spacetime satisfying the vacuum Einstein field equations
 \begin{equation} \tilde{R}_{ab}=0.  \nonumber
 \end{equation} Further, denote by $\bmg$ an \emph{unphysical metric}
conformally related to $\tilde{\bm{g}}$ via the relation
 \begin{equation} g_{ab} = \Xi^2 \tilde{g}_{ab}, \nonumber
 \end{equation} with $\Xi$ a suitable conformal factor.
 
 \subsubsection{Connections and curvature.}  In the following we will
denote by $\nabla_a$ and $\tilde{\nabla}_a$ the Levi-Civita
connections of $\bm{g}$ and $\tilde{\bm{g}}$, respectively. In
addition to the latter, we will consider a \emph{Weyl connection}
$\hat{\nabla}_a$, that is, a torsion-free connection such that
 \begin{equation} \hat{\nabla}_a g_{bc} =-2f_a g_{bc}.  \nonumber
 \end{equation} It follows from the previous discussion that the
transformation between connections $\nabla_a$ and $\hat{\nabla}_a$ is
given by
 \begin{equation} \hat{\nabla}_a - \nabla_a v^b = S_{ac}{}^{bd}f_d
v^c, \qquad S_{ac}{}^{bd}\equiv \delta_a{}^b\delta_c{}^d
+\delta_a{}^d\delta_c{}^b - g_{ac}g^{bd}, \nonumber
 \end{equation} where $f_a$ is a fixed smooth covector and $v^a$ is an
arbitrary vector. Moreover, observing that
 \begin{equation} \nabla_a v^b -\tilde{\nabla}_a v^b
=S_{ac}{}^{bd}(\Xi^{-1}\nabla_a\Xi) v^c, \nonumber
 \end{equation} it follows that
 \begin{equation} \hat{\nabla}_a v^b - \tilde{\nabla}_a v^b =
S_{ac}{}^{bd} \beta_d v^c, \qquad \beta_d \equiv f_d +
\Xi^{-1}\nabla_d\Xi.  \nonumber
 \end{equation} We make the following definition:
 \begin{equation} d_a \equiv \Xi f_a +\nabla_a \Xi.  \nonumber
 \end{equation}

 Throughout the rest of this article, $\hat{R}^a{}_{bcd}$ and
$\hat{L}_{ab}$ will denote, respectively, the \emph{Riemann and
Schouten tensors} of $\hat{\nabla}_a$. Generically, one has that
$\hat{L}_{ab}\neq \hat{L}_{ba}$. The Riemann tensor admits the
decomposition
 \begin{equation} \hat{R}^c{}_{dab} = 2 S_{d[a]}{}^{ce}\hat{L}_{b]e}+
C^c{}_{dab}, \nonumber
 \end{equation} with $C^c{}_{dab}$ the conformally invariant
\emph{Weyl tensor}. It will be useful to consider the (vanishing)
torsion tensor $\hat{\Sigma}_a{}^c{}_b$ of $\hat{\nabla}_a$. Finally,
define the \emph{rescaled Weyl tensor} $d^c{}_{dab}$ via the relation
\begin{equation} d^c{}_{dab}\equiv \Xi^{-1} C^c{}_{dab}.  \nonumber
\end{equation}

\subsubsection{A frame formalism.} Let $\{ \bm{e}_{\bm{a}} \}$,
$\bm{a}=0,\,1,\,2,\,3,\,$ denote an orthonormal frame with respect to
the metric $\bm{g}$. Also, let $\{ \bm{\omega}^{\bm{b}} \}$ be the
associated coframe.  Given a vector $v^a$, its components with respect
to the frame $\{ \bm{e}_{\bm{a}} \}$ will be written as
$v^{\bm{a}}$. We denote by $\Gamma_{\bm{a}}{}^{\bm{b}}{}_{\bm{c}}$ and
$\hat{\Gamma}_{\bm{a}}{}^{\bm{b}}{}_{\bm{c}}$, respectively, the
connection coefficients of $\nabla_a$ and $\hat{\nabla}_a$ with
respect to $\bm{e}_{\bm{a}}$. One has then that
\begin{equation} \hat{\Gamma}_{\bm{a}}{}^{\bm{b}}{}_{\bm{c}} =
\Gamma_{\bm{a}}{}^{\bm{b}}{}_{\bm{c}} +
S_{\bm{a}\bm{b}}{}^{\bm{c}\bm{d}}f_{\bm{d}}, \qquad f =\frac{1}{4}
\hat{\Gamma}_{\bm{a}}{}^{\bm{b}}{}_{\bm{b}}.  \nonumber
\end{equation} Given a tensor $T^{\bm{b}}{}_{\bm{b}}$ and letting
$\partial_{\bm{a}}\equiv e_{\bm{a}}{}^\mu\partial_\mu$ one has that
\begin{eqnarray*} && \nabla_{\bm{a}}T^{\bm{b}}{}_{\bm{c}} \equiv
e_{\bm{a}}{}^a \omega^{\bm{b}}{}_b \omega^{\bm{c}}{}_c (\nabla_a
T^b{}_c)\\ && \phantom{\nabla_{\bm{a}}T^{\bm{b}}{}_{\bm{c}}} =
\partial_{\bm{a}}T^{\bm{b}}{}_{\bm{c}} +
\Gamma_{\bm{a}}{}^{\bm{b}}{}_{\bm{d}}T^{\bm{d}}{}_{\bm{c}} -
\Gamma_{\bm{a}}{}^{\bm{d}}{}_{\bm{c}}T^{\bm{b}}{}_{\bm{d}}.
\end{eqnarray*}

\subsection{The frame version of the extended conformal Einstein field
equations} To express the extended conformal Einstein field equations
in a concise manner we introduce the following \emph{zero-quantities}:
\begin{eqnarray*} &&
\hat{\Sigma}{}_{\bm{a}}{}^{\bm{c}}{}_{\bm{b}}\bm{e}_{\bm{c}} \equiv
[\bm{e}_{\bm{a}},\bm{e}_{\bm{b}}]- \Big(
\hat{\Gamma}{}_{\bm{a}}{}^{\bm{c}}{}_{\bm{b}}-\hat{\Gamma}_{\bm{b}}{}^{\bm{c}}{}_{\bm{a}}
\Big) \bm{e}_{\bm{c}}, \\ &&
\hat{\Xi}{}^{\bm{c}}{}_{\bm{d}\bm{a}\bm{b}} \equiv
\hat{P}{}^{\bm{c}}{}_{\bm{d}\bm{a}\bm{b}}
-\hat{\rho}{}^{\bm{c}}{}_{\bm{d}\bm{a}\bm{b}}, \\ &&
\hat{\Delta}_{\bm{c}\bm{d}\bm{b}} \equiv \hat{\nabla}_\bmc
\hat{L}_{\bmd\bmb}-\hat{\nabla}_\bmd \hat{L}_{\bmc\bmd}-d_\bma
d^\bma{}_{\bmb\bmc\bmd}, \\ && \Lambda_{\bmb\bmc\bmd} \equiv
\nabla_\bma d^\bma{}_{\bmb\bmc\bmd},
\end{eqnarray*} where the components of the \emph{geometric curvature}
$\hat{P}{}^{\bm{c}}{}_{\bm{d}\bm{a}\bm{b}}$ and the \emph{algebraic
curvature} $\hat{\rho}{}^{\bm{c}}{}_{\bm{d}\bm{a}\bm{b}}$ are given,
respectively, by
\begin{eqnarray*} && \hat{P}{}^{\bm{c}}{}_{\bm{d}\bm{a}\bm{b}} \equiv
\partial_\bma \hat{\Gamma}{}_\bmb{}^\bmc{}_\bmd - \partial_\bmb
\hat{\Gamma}{}_\bma{}^\bmc{}_\bmd + \hat{\Gamma}{}_\bmf{}^\bmc{}_\bmd
\big(\hat{\Gamma}{}_\bmb{}^\bmf{}_\bma
-\hat{\Gamma}_\bma{}^\bmf{}_\bmb\big) + \hat{
\Gamma}{}_\bmb{}^\bmf{}_\bmd\hat{\Gamma}{}_\bma{}^\bmc{}_\bmf
-\hat{\Gamma}{}_\bma{}^\bmf{}_\bmd\hat{\Gamma}{}_\bmb{}^\bmc{}_\bmf,
\\ && \hat{\rho}{}^\bmc{}_{\bmd\bma\bmb}\equiv \Xi
d^\bmc{}_{\bmd\bma\bmb}+2
S_{\bmd[\bma]}{}^{\bmc\bme}\hat{L}_{\bmb]\bme}.
\end{eqnarray*} The \emph{extended conformal Einstein field equations}
can be expressed in terms of the above zero-quantities as the
conditions
\begin{equation} \hat{\Sigma}{}_\bma{}^\bmc{}_\bmb=0, \qquad
\hat{\Xi}{}^\bmc{}_{\bmd\bma\bmb}=0, \qquad
\hat{\Delta}_{\bmc\bmd\bmb}=0, \qquad \hat{\Lambda}_{\bmb\bmc\bmd}=0.
\nonumber
\end{equation} In the above equations the fields $\Xi$ and $d_\bma$
are regarded as \emph{conformal gauge fields} which are determined by
means of some supplementary conditions. In the particular case of this
article, these aforementioned conformal gauge fields will be
determined through a gauge associated to conformal geodesics. In order
to account for this it is convenient to define
\begin{eqnarray*} && \delta_\bma \equiv d_\bma -\Xi f_\bma
-\nabla_\bma \Xi, \\ && \gamma_{\bma\bmb}\equiv
\hat{L}_{\bma\bmb}-\hat{\nabla}_\bma(\Xi^{-1}d_\bmb)-\frac{1}{2}\Xi^{-1}S_{\bma\bmb}{}^{\bmc\bmd}d_\bmc
d_\bmd, \\ && \varsigma_{\bma\bmb} \equiv
\hat{L}_{[\bma\bmb]}=\hat{\nabla}_{[\bma}f_{\bmb]}.
\end{eqnarray*} The conditions
\begin{equation} \delta_\bma=0, \qquad \gamma_{\bma\bmb}=0, \qquad
\varsigma_{\bma\bmb}=0, \nonumber
\end{equation} are the \emph{supplementary conditions}. They play a
key role in relating the Einstein field equations to the extended
conformal Einstein field equations and also in the analysis of the
propagation of the constraints.

\subsection{The conformal Gau\ss ian gauge} As already mentioned, we
will be working on a gauge based on the properties of conformal
geodesics. To this end, in the following we assume that we have a
region $\mathcal{U}\subset \mathcal{M}$ which can be covered by a
congruence of non-intersecting conformal geodesics.

\medskip Recall that a \emph{conformal geodesic} on
$(\tilde{\mathcal{M}},\tilde{\bmg})$ is a pair
$(x(\tau),\tilde{\beta}_a(\tau))$ consisting of a curve $x(\tau)$ in
$\tilde{\mathcal{M}}$ with parameter $\tau\in I$ with
$I\subset\mathbb{R}$ an interval, and a covector
$\tilde{\beta}_a(\tau)$ along the curve which satisfy the equations
\begin{eqnarray*} && \dot{x}^c\tilde{\nabla}_c \dot{x}^a=-2
\tilde{\beta}_c \dot{x}^c \dot{x}^a +
\tilde{g}_{bc}\dot{x}^b\dot{x}^c\tilde{\beta}^a, \\ &&
\dot{x}^c\tilde{\nabla}_c \tilde{\beta}_a = \tilde{\beta}_c\dot{x}^c
\tilde{\beta}_a -
\frac{1}{2}\tilde{g}^{bc}\tilde{\beta}_b\tilde{\beta}_c \dot{x}_a +
\tilde{L}_{ca}\dot{x}^a.
\end{eqnarray*} When discussing conformal geodesics is natural to
consider, in addition, a frame $\{ \bme_\bma \}$ \emph{Weyl
propagated} along $x(\tau)$ according to
\begin{equation} \dot{x}^c\tilde{\nabla}_c e_\bma{}^a = -
(\tilde{\beta}_c e_\bma{}^c) \dot{x}^a -(\tilde{\beta}_c\dot{x}^c)
e_\bma{}^a + (\tilde{g}_{cd}e_\bma{}^c \dot{x}^d) \tilde{\beta}^a.
\nonumber
\end{equation}

The main motivation behind the use of a gauge associated to conformal
geodesics, in particular, timelike conformal geodesics, is the
following:

\medskip
\noindent
\begin{proposition}
\label{Proposition:CanonicalConformalFactor} Let
$(\tilde{\mathcal{M}},\tilde{\bmg})$ denote a spacetime satisfying the
vacuum Einstein field equations. Suppose that
$(x(\tau),\tilde{\beta}_a)$ is a solution to the conformal geodesic
equations and that $\{\bme_\bma\}$ is a Weyl propagated orthonormal
frame with respect to $\tilde{\bmg}$. Let $\Theta$ denote a conformal
factor such that
\begin{equation} g_{ab} =\Theta^2 \tilde{g}_{ab}, \qquad
g_{ab}\dot{x}^a \dot{x}^b =1, \nonumber
\end{equation} then one has that
\begin{equation} \Theta(\tau) =\Theta_\star
+\dot{\Theta}_\star(\tau-\tau_\star) +\frac{1}{2}\ddot{\Theta}_\star
(\tau-\tau_\star)^2, \nonumber
\end{equation} where the coefficients $\Theta_\star$,
$\dot{\Theta}_\star$ and $\ddot{\Theta}_\star$ are constant along a
conformal geodesic and satisfy
\begin{equation} \dot{\Theta}_\star = \tilde{\beta}_{a\star}
\dot{x}^a_\star \Theta_\star, \qquad \Theta_\star \ddot{\Theta}_\star
=\frac{1}{2}\tilde{g}^{bc}\tilde{\beta}_b\tilde{\beta}_c.
\label{CanonicalConformalFactor}
\end{equation} Furthermore, along each conformal geodesic the
components of $d_a$ satisfy
\begin{equation} d_\bmzero =\dot{\Theta}_\star, \qquad d_\bmi =
\Theta_\star \tilde{\beta}_{\bmi\star}.
\end{equation}
\end{proposition} A proof of the above result can be found in
\cite{kroon2017conformal}.

\subsection{The hyperbolic reduction of the XCFE in the conformal
Gaussian gauge} \label{Subsection:HyperbolicReduction} For further
reference, in this subsection we provide a brief discussion of the
hyperbolic reduction of the extended conformal field equations in the
conformal Gaussian gauge. For concreteness, here we focus on the frame
version of the equations. Note, however, that a spinorial version
thereof can be readily be available.

\subsubsection{Gauge conditions.} In order to proceed with the
hyperbolic reduction of the XCFE, in the following we assume we have a
region $\mathcal{U}$ of the spacetime
$(\tilde{\mathcal{M}},\tilde{\bmg})$ which is covered by a
non-intersecting congruence of conformal geodesics. The vector tangent
to the curves in the congrience will be denoted by $\tau^a$. As
mentioned in the previous subsection, a conformal representative
$\bmg$ of the conformal class $[\tilde{\bmg}]$ is singled out by the
requirement $g_{ab}\tau^a\tau^b=1$, so that $g_{ab} =\Theta^2
\tilde{g}_{ab}$ where $\Theta$ is the conformal factor given by
equation \eqref{CanonicalConformalFactor}. This conformal factor is
completely determined by the the coefficients $\Theta_\star$,
$\dot{\Theta}_\star$, $\ddot{\Theta}_\star$ specified, say, on an
initial hypersurface $\mathcal{S}_\star$ ---see Proposition
\ref{Proposition:CanonicalConformalFactor}.

In the following we consider a $\bmg$-orthogonal frame $\{ \bme_\bma
\}$ which is Weyl propagated along the conformal geodesics and such
that $e_\bmzero{}^a=\tau^a $. Now, to every congruence of conformal
geodesics one can associate a Weyl connection $\hat{\nabla}_a$. It
follows from the conformal geodesic equations that this connection
satisfies
\begin{equation} \tau^a\hat{\nabla}_a \bme_\bma=0, \qquad
\hat{L}_{ab}\tau^a=0.  \nonumber
\end{equation} In terms of frame components, the above conditions can
be rewritten as
\begin{equation} \hat{\Gamma}_\bmzero{}^\bma{}_\bmb=0, \qquad
\hat{L}_{\bmzero\bma}=0.  \nonumber
\end{equation} In particular, the covector $f_a$ defining the Weyl
connection satisfies
\begin{equation} f_\bmzero=0.  \nonumber
\end{equation} It is possible to refine the gauge further by choosing
the parameter $\tau$ of the conformal geodesics as the time
coordinate. In this way one gets the additional gauge condition
\begin{equation} \bme_\bmzero={\bm\partial}_\tau, \qquad
e_\bmzero{}^\mu =\delta_0{}^\mu.  \nonumber
\end{equation}

\subsubsection{Evolution equations in tensorial form.} The required
evolution equations for the frame components, connection coefficients
and components of the Schouten tensor are obtained from the conditions
\begin{equation} \hat{\Sigma}_\bmzero{}^\bmb{}_\bmc =0, \qquad
\hat{\Xi}^\bmc_{\bmd\bmzero\bmb}=0, \qquad
\hat{\Delta}_{\bmzero\bmb\bmc}=0.  \nonumber
\end{equation} Expanding these equations in terms of the gauge
conditions described in the previous paragraph one obtains the
evolution equations
\begin{subequations}
\begin{eqnarray} && \partial_\tau e_\bmb{}^\mu =
-\hat{\Gamma}_\bmb{}^\bmf{}_\bmzero e_\bmf{}^\mu, \label{FrameP1}\\ &&
\partial_\tau \hat{\Gamma}_\bmb{}^\bmc{}_\bmd =
-\hat{\Gamma}_\bmf{}^\bmc{}_\bmd\hat{\Gamma}_\bmb{}^\bmf{}_\bmzero
+\delta_\bmzero{}^\bmc\hat{L}_{\bmb\bmd}+\delta_\bmd{}^\bmc\hat{L}_{\bmb\bmzero}-\eta_{\bmzero\bmd}\eta^{\bmf\bmc}\hat{L}_{\bmb\bmf}+\Theta
d^\bmc{}_{\bmd\bmzero\bmb}, \\ && \partial_\tau\hat{L}_{\bmb\bmc} =
-\hat{\Gamma}_\bmb{}^\bmf{}_\bmzero \hat{L}_{\bmf\bmc}+d_\bmf
d^\bmf{}_{\bmc\bmzero\bmb}.
\end{eqnarray}
\end{subequations}

Observe that these equations contain derivatives only in the $\tau$
direction ---in other words, they are \emph{transport equations}.

Evolution equations for the components of the Weyl tensor can be
obtained by expanding the conditions
\begin{equation} \Lambda^*_{(\bmb|\bmzero|\bmd)}=0, \qquad
\Lambda_{(\bmb|\bmzero|\bmd)}=0.  \nonumber
\end{equation} We do not expand these conditions. In the sequel we
will use a spinorial formulation of the equations as it more clearly
brings to the fore the structural properties of the equations.

\subsubsection{The spinorial form of the evolution equations.} The
spinorial counterpart of the fields appearing in the frame version of
the XCFE is given by the components
\begin{equation} e_{\bmA\bmA'}{}^\mu, \qquad
\hat{\Gamma}_{\bmA\bmA'\bmC\bmD}, \quad \hat{L}_{\bmA\bmA'\bmB\bmB'},
\quad \phi_{\bmA\bmB\bmC\bmD} \nonumber
\end{equation} with respect to a spinor dyad $\{ \epsilon_\bmA{}^A\}$
adapted to the spinorial counterpart $\tau^{AA'}$ of the tangent
vector $\tau^a$ to the congruence of conformal geodesics fixing our
gauge ---that is, one has that
\begin{equation} \tau^{AA'} =\epsilon_\bmzero{}^A
\epsilon_{\bmzero'}{}^{A'}
+\epsilon_\bmone{}^A\epsilon_{\bmone'}{}^{A'}.  \nonumber
\end{equation} Observe that $\tau_{AA'}\tau^{AA'}=2$. In terms of the
usual $\{o^A,\,\iota^A\}$-notation, one has that
$\epsilon_\bmzero{}^A=o^A$, $\epsilon_\bmone{}^A=\iota^A$.

The spinorial version of the conditions associated to the conformal
Gaussian gauge is given by
\begin{equation}
\tau^{\bmA\bmA'}\bme_{\bmA\bmA'}=\sqrt{2}\partial_\tau, \qquad \tau^{
\bmA\bmA'}\hat{\Gamma}_{\bmA\bmA'\bmB\bmC}=0, \qquad
\tau^{\bmA\bmA'}\hat{L}_{\bmA\bmA'\bmB\bmB'}=0.  \nonumber
\end{equation} We also note that the Weyl connection spin coefficients
can be expressed in terms of the unphysical Levi-Civita connection
coefficients as
\begin{equation} \hat{\Gamma}_{\bmC\bmC'\bmA\bmB} =
\Gamma_{\bmC\bmC'\bmA\bmB} -\epsilon_{\bmA\bmC}f_{\bmB\bmC'}.
\nonumber
\end{equation} From the gauge conditions one further finds that
\begin{equation} \tau^{\bmC\bmC'}\Gamma_{\bmC\bmC'\bmA\bmB}
=-\tau_\bmA{}^{\bmC'}f_{\bmB\bmC'}.
\end{equation}

It follows from the previous discussion that in the conformal Gau\ss
ian gauge, one has that the Weyl connection spin connection
coefficients can be expressed fully in terms of the Levi-Civita
connection spin coefficients and vice versa. The spinor $f_{AA'}$
encodes the acceleration of the congruence of conformal geodesics.

The spinorial version of the evolution equations is then given by
\begin{eqnarray*} && \sqrt{2}\partial_\tau e_{\bmA\bmA'}{}^\mu =
-\Big( \hat{\Gamma}_{\bmA\bmA'}{}^\bmQ{}_\bmB \tau^{\bmB\bmQ'} +
\bar{\hat{\Gamma}}{}_{\bmA'\bmA}{}^{\bmQ'}{}_{\bmB'} \tau^{\bmQ\bmB'}
\Big) e_{\bmQ\bmQ'}{}^\mu, \\ && \sqrt{2}\partial_\tau
\hat{\Gamma}_{\bmA\bmA'}{}^\bmB{}_\bmC =-\Big(
\hat{\Gamma}_{\bmA\bmA'}{}^\bmP{}_\bmQ
\hat{\Gamma}_{\bmP\bmQ'}{}^\bmB{}_\bmC
+\bar{\hat{\Gamma}}{}_{\bmA\bmA'}{}^{\bmP'}{}_{\bmQ'}\Gamma_{\bmQ\bmP'}{}^\bmB{}_\bmC
\Big)\tau^{\bmQ\bmQ'}\\ && \hspace{4cm}+
\hat{L}_{\bmA\bmA'\bmC\bmQ'}\tau^{\bmB\bmQ'}+ \Theta
\phi^\bmB{}_{\bmC\bmQ\bmA}\tau^Q{}_{\bmA',} \\ &&
\sqrt{2}\partial_\tau \hat{L}_{\bmA\bmA'\bmB\bmB'}= -\Big(
\hat{\Gamma}_{\bmA\bmA'}{}^\bmP{}_\bmQ \hat{L}_{\bmP\bmQ'\bmB\bmB'} +
\bar{\hat{\Gamma}}{}_{\bmA'\bmA}{}^{\bmP'}{}_{\bmQ'}\hat{L}_{\bmQ\bmP'\bmB\bmB'}
\Big)\tau^{\bmQ\bmQ'} \\ && \hspace{4cm}- d^{\bmP\bmP'}\Big(
\phi_{\bmP\bmA\bmQ\bmB}\epsilon_{\bmP'\bmB'}\tau^Q{}_{A'} +
\hat{\phi}_{\bmP'\bmA'\bmQ'\bmB'}\epsilon_{\bmP\bmB}\tau_\bmA{}^{\bmQ'}
\Big).
\end{eqnarray*} The evolution equations associated to the Bianchi
equations are obtained from a \emph{space spinor decomposition} of the
spinor Bianchi equation
\begin{equation} \nabla^A{}_{A'}\phi_{ABCD}=0.
\label{SpinorBianchi}
\end{equation} The details of this decomposition are given in the next
subsection.

$\mathcal{D}\equiv \tau^{AA'}\nabla_{AA'}$ and \emph{Sen connection}
$\mathcal{D}_{AB}\equiv \tau_{(A}{}^{A'}\nabla_{B)A'}$ associated to
the Hermitian spinor $\tau^{AA'}$ in the form

\subsubsection{Space spinor
split.} \label{Subsubsection:SpaceSpinorSplit} A more detailed version
of the evolution equations can be obtained by means of the \emph{space
spinor formalism} ---see e.g. \cite{kroon2017conformal}, Chapter 4;
also \cite{sommers1980space}. Following the general strategy behind
this formalism, one defines the \emph{unprimed spinors}
\begin{eqnarray*} & \hat{\Gamma}_{\bmA\bmB\bmC\bmD}\equiv
\tau_\bmB{}^{\bmA'}\hat{\Gamma}_{\bmA\bmA'\bmC\bmD}, \qquad
\Gamma_{\bmA\bmB\bmC\bmD} \equiv
\tau_\bmB{}^{\bmA'}\Gamma_{\bmA\bmA'\bmC\bmD}, \qquad
f_{\bmA\bmB}\equiv \tau_\bmB{}^{\bmA'}f_{\bmA\bmA'},& \\ &
\Theta_{\bmA\bmB\bmC\bmD}\equiv
\tau_\bmB{}^{\bmA'}\tau_\bmD{}^{\bmC'}\hat{L}_{\bmA\bmA'\bmC\bmC'}, &
\end{eqnarray*} so that, in particular, one has that
\begin{equation} \hat{\Gamma}_{\bmA\bmB\bmC\bmD}
=\Gamma_{\bmA\bmB\bmC\bmD} -\epsilon_{\bmC\bmA}f_{\bmD\bmB}.
\end{equation} The gauge conditions imply the symmetry conditions
\begin{equation} f_{\bmA\bmB}=f_{(\bmA\bmB)}, \qquad
\Gamma_\bmQ{}^\bmQ{}_{\bmA\bmB}=-f_{\bmA\bmB},\qquad
\Theta_\bmQ{}^\bmQ{}_{\bmA\bmB}=0.  \nonumber
\end{equation} The Hermitian conjugation associated to the spinor
$\tau^{AA'}$ can be used to split the spin connection coefficients
into real and imaginary parts through the definitions
\begin{eqnarray*} \chi_{\bmA\bmB\bmC\bmD} \equiv -\frac{1}{\sqrt{2}}
\big( \Gamma_{\bmA\bmB\bmC\bmD}+\hat{\Gamma}_{\bmA\bmB\bmC\bmD}\big),
\qquad \xi_{\bmA\bmB\bmC\bmD}\equiv \frac{1}{\sqrt{2}}\big(
\Gamma_{\bmA\bmB\bmC\bmD} - \hat{\Gamma}_{\bmA\bmB\bmC\bmD}\big),
\end{eqnarray*} so that
\begin{eqnarray*} && \Gamma_{\bmA\bmB\bmC\bmD}
=\frac{1}{\sqrt{2}}\big(
\xi_{\bmA\bmB\bmC\bmD}-\chi_{\bmA\bmB\bmC\bmD} \big)\\ &&
\phantom{\Gamma_{\bmA\bmB\bmC\bmD}}=\frac{1}{\sqrt{2}}\big(
\xi_{\bmA\bmB\bmC\bmD}-\chi_{(\bmA\bmB)\bmC\bmD} \big)
-\frac{1}{2}\epsilon_{\bmA\bmB}f_{\bmC\bmD}.
\end{eqnarray*} In addition, from the gauge conditions, one has that
the spinor $\Theta_{\bmA\bmB\bmC\bmD}$ can be decomposed as
\begin{equation} \Theta_{\bmA\bmB\bmC\bmD}
=\Theta_{\bmA\bmB(\bmC\bmD)}+\frac{1}{2}\epsilon_{\bmC\bmD}\Theta_{\bmA\bmB\bmQ}{}^\bmQ.
\nonumber
\end{equation} Consistent with the above, one can define the
\emph{electric} and \emph{magnetic parts} of the spinor
$\phi_{\bmA\bmB\bmC\bmD}$ as
\begin{equation} \eta_{\bmA\bmB\bmC\bmD}\equiv \frac{1}{2}\big(
\phi_{\bmA\bmB\bmC\bmD}+\hat{\phi}_{\bmA\bmB\bmC\bmD}\big), \qquad
\mu_{\bmA\bmB\bmC\bmD}\equiv -\frac{\mathrm{i}}{2}\big(
\phi_{\bmA\bmB\bmC\bmD}-\hat{\phi}_{\bmA\bmB\bmC\bmD}\big).  \nonumber
\end{equation} A further calculation then yields the following final
form of the transport equations:
\begin{subequations}
\begin{eqnarray} &&\partial_\tau
e_{\bmA\bmB}{}^0=-\chi_{(\bmA\bmB)}{}^{\bmP\bmQ}
e_{\bmP\bmQ}{}^0-f_{\bmA\bmB}, \label{p1} \\ &&\partial_\tau
e_{\bmA\bmB}{}^\alpha=-\chi_{(\bmA\bmB)}{}^{\bmP\bmQ}
e_{\bmP\bmQ}{}^\alpha, \label{p2}\\ &&\partial_\tau
\xi_{\bmA\bmB\bmC\bmD}=-\chi_{(\bmA\bmB)}{}^{\bmP\bmQ}\xi_{\bmP\bmQ\bmC\bmD}+\displaystyle\frac{1}{\sqrt{2}}(\epsilon_{\bmA\bmC}\chi_{(\bmB\bmD)\bmP\bmQ}+\epsilon_{\bmB\bmD}\chi_{(\bmA\bmC)\bmP\bmQ})f^{\bmP\bmQ}
\nonumber\\ &&\hspace{3cm} -\sqrt{2}\chi_{(\bmA\bmB)(\bmC}{}^\bmE
f_{\bmD)\bmE}-\displaystyle\frac{1}{2}(\epsilon_{\bmA\bmC}
\Theta_{\bmB\bmD\bmQ}{}^\bmQ
+\epsilon_{\bmB\bmD}\Theta_{\bmA\bmC\bmQ}{}^\bmQ)\nonumber \\ &&
\hspace{3cm} -\mbox{i}\Theta\mu_{\bmA\bmB\bmC\bmD}, \label{p3} \\
&&\partial_\tau
f_{\bmA\bmB}=-\chi_{(\bmA\bmB)}{}^{\bmP\bmQ}f_{\bmP\bmQ}+\displaystyle\frac{1}{\sqrt{2}}\Theta_{\bmA\bmB\bmQ}{}^\bmQ, \label{p4}
\\ &&\partial_\tau
\chi_{(\bmA\bmB)\bmC\bmD}=-\chi_{(\bmA\bmB)}{}^{\bmP\bmQ}\chi_{\bmP\bmQ\bmC\bmD}-\Theta_{\bmA\bmB(\bmC\bmD)}+\Theta\eta_{\bmA\bmB\bmC\bmD}, \label{p5}
\\
&&\partial_\tau\Theta_{\bmC\bmD(\bmA\bmB)}=-\chi_{(\bmC\bmD)}{}^{\bmP\bmQ}\Theta_{\bmP\bmQ(\bmA\bmB)}-\partial_\tau\Theta\eta_{\bmA\bmB\bmC\bmD}\nonumber\\
&&
\hspace{3cm}+\mbox{i}\sqrt{2}d^\bmP{}_{(\bmA}\mu_{\bmB)\bmC\bmD\bmP}, \label{p6}
\\ &&\partial_\tau \Theta_{\bmA\bmB\bmQ}{}^\bmQ
=-\chi_{(\bmA\bmB)}{}^{\bmE\bmF}\Theta_{\bmE\bmF\bmQ}{}^\bmQ+\sqrt{2}d^{\bmP\bmQ}\eta_{\bmA\bmB\bmP\bmQ}. \label{p7}
\end{eqnarray}
\end{subequations}

For future use define the spinorial Bianchi zero-quantity
\begin{equation} \Lambda_{\bmA'\bmB\bmC\bmD}\equiv
\nabla^\bmQ{}_{\bmA'}\phi_{\bmB\bmC\bmD\bmQ} \nonumber
\end{equation} and its space spinor version
\begin{equation} \Lambda_{\bmA\bmB\bmC\bmD}\equiv
\tau_\bmA{}^{\bmA'}\Lambda_{\bmA'\bmB\bmC\bmD}, \qquad
\Lambda_{\bmA\bmB\bmC\bmD}=\Lambda_{\bmA(\bmB\bmC\bmD)} \nonumber
\end{equation} ---cfr. equation \eqref{SpinorBianchi}.  In particular,
one has the decomposition
\begin{equation} \Lambda_{\bmA\bmB\bmC\bmD}
=\Lambda_{(\bmA\bmB\bmC\bmD)}-\frac{3}{4}\epsilon_{\bmA(\bmB}\Lambda_{\bmC\bmD)},
\qquad \Lambda_{\bmC\bmD}\equiv \Lambda^\bmQ{}_{\bmQ\bmC\bmD}.
\nonumber
\end{equation} A calculation the readily allows to verify that the
conditions
\begin{equation} \Lambda_{(\bmA\bmB\bmC\bmD)}=0, \qquad
\Lambda_{\bmA\bmB}=0, \nonumber
\end{equation} imply, respectively, the equations
\begin{subequations}
\begin{eqnarray} && \mathcal{D}\phi_{\bmA\bmB\bmC\bmD} -
2\mathcal{D}_{(\bmA}{}^\bmQ\phi_{\bmB\bmC\bmD)\bmQ}=0, \label{StandardBianchiEvolution}\\
&&
\mathcal{D}^{\bmA\bmB}\phi_{\bmA\bmB\bmC\bmD}=0,\label{BianchiConstraint}
\end{eqnarray}
\end{subequations} where $\mathcal{D} := \tau^a\nabla_a$ is the
covariant derivative along the conformal geodesics and $D_{AB} :=
\tau_{(B}{}^{A'}\nabla_{A)A'}$ is the Sen connection of $\nabla_{AA'}$
induced by $\tau^{AA'}$.  Equation \eqref{StandardBianchiEvolution}
expanded in terms of the components
\begin{equation} \phi_0 \equiv \phi_{\bmzero\bmzero\bmzero\bmzero},
\quad \phi_1\equiv \phi_{\bmzero\bmzero\bmzero\bmone}, \quad
\phi_2\equiv \phi_{\bmzero\bmzero\bmone\bmone}, \quad \phi_3 \equiv
\phi_{\bmzero\bmone\bmone\bmone}, \quad \phi_4 \equiv
\phi_{\bmone\bmone\bmone\bmone}, \nonumber
\end{equation} gives rise to the so-called \emph{standard evolution
system}, whereas \eqref{BianchiConstraint} encodes the Bianchi
constraints.

\medskip In the analysis of the boundary value problems, the use of an
alternative evolution system, the \emph{boundary adapted system}, is a
significant advantage. This is obtained from the combinations
\begin{subequations}
\begin{eqnarray} & -2\Lambda_{(\bmzero\bmzero\bmzero\bmzero)}=0, \quad
-2 \Lambda_{(\bmzero\bmzero\bmzero\bmone)} -
\displaystyle\frac{1}{2}C_{\bmzero\bmzero}=0, \quad -2
\Lambda_{(\bmzero\bmzero\bmone\bmone)}=0,& \label{HR:BoundaryAdaptedSystem1}\\
& -2 \Lambda_{(\bmzero\bmone\bmone\bmone)} +
\displaystyle\frac{1}{2}C_{\bmone\bmone}=0, \quad -2
\Lambda_{\bmone\bmone\bmone\bmone}=0,
& \label{HR:BoundaryAdaptedSystem2}
\end{eqnarray}
\end{subequations}

Explicitly, the boundary adapted system takes the form
\begin{eqnarray*} && (\sqrt{2}+ 2 e_{\bmzero\bmone}{}^0) \partial_\tau
\phi_0 -2e_{\bmone\bmone}{}^0 \partial_\tau \phi_1 + 2
e_{\bmzero\bmone}{}^\alpha\partial_\alpha \phi_0 -2
e_{\bmone\bmone}{}^\alpha \partial_\alpha \phi_1, \\ && \hspace{1cm} =
-6 \Gamma_{\bmone\bmone\bmone\bmone} \phi_2 +
(4\Gamma_{\bmone\bmone\bmone\bmzero} + 8
\Gamma_{\bmzero\bmone\bmone\bmone})\phi_1 +
(2\Gamma_{\bmone\bmone\bmzero\bmzero}-8
\Gamma_{\bmzero\bmone\bmzero\bmone})\phi_0, \\ &&
\sqrt{2}\partial_\tau \phi_1 -e_{\bmone\bmone}{}^0 \partial_\tau
\phi_2 + e_{\bmzero\bmzero}{}^0 \partial_\tau \phi_0 -
e_{\bmone\bmone}{}^\alpha \partial_\alpha \phi_2 +
e_{\bmzero\bmzero}{}^\alpha\partial_\alpha \phi_4 \\ && \hspace{1cm} =
-2 \Gamma_{\bmone\bmone\bmone\bmone} \phi_3 -3 f_{\bmone\bmone} \phi_2
+ (2\Gamma_{\bmone\bmone\bmzero\bmzero} + 4
\Gamma_{\bmzero\bmzero\bmone\bmone}+2 f_{\bmzero\bmone})\phi_1
-(4\Gamma_{\bmzero\bmzero\bmzero\bmone}-f_{\bmzero\bmzero})\phi_0, \\
&& \sqrt{2} \partial_\tau \phi_2 - e_{\bmone\bmone}{}^0 \partial_\tau
\phi_3 + e_{\bmzero\bmzero}{}^0 \partial_\tau \phi_1 -
e_{\bmone\bmone}{}^\alpha \partial_\alpha\phi_3 +
e_{\bmzero\bmzero}{}^\alpha \partial_\alpha \phi_1 \\ && \hspace{1cm}
= -\Gamma_{\bmone\bmone\bmone\bmone} \phi_4 -2
(\Gamma_{\bmone\bmone\bmzero\bmone}+ f_{\bmone\bmone} )\phi_3 + 3
(\Gamma_{\bmzero\bmzero\bmone\bmone}+
\Gamma_{\bmone\bmone\bmzero\bmzero})\phi_2\\ && \hspace{2cm}
-2(\Gamma_{\bmzero\bmzero\bmzero\bmone} - f_{\bmzero\bmzero})\phi_1
-\Gamma_{\bmzero\bmzero\bmzero\bmzero}\phi_0, \\ && \sqrt{2}
\partial_\tau \phi_3 -e_{\bmone\bmone}{}^0 \partial_\tau \phi_4
+e_{\bmzero\bmzero}{}^0 \partial_\tau \phi_2 -
e_{\bmone\bmone}{}^\alpha \partial_\alpha\phi_4 +
e_{\bmzero\bmzero}{}^\alpha\partial_\alpha\phi_2 \\ && \hspace{1cm} =
-(4 \Gamma_{\bmone\bmone\bmone\bmzero} + f_{\bmone\bmone})\phi_2 +
(2\Gamma_{\bmzero\bmzero\bmone\bmone}+ 4
\Gamma_{\bmone\bmone\bmzero\bmzero}-2 f_{\bmzero\bmone})\phi_3 +
3f_{\bmzero\bmzero} \phi_2 - 2\Gamma_{\bmzero\bmzero\bmzero\bmzero}
\phi_1, \\ && (\sqrt{2} -2 e_{\bmzero\bmone}{}^0) \partial_\tau \phi_4
+ 2 e_{\bmzero\bmzero}{}^0 \partial_\tau \phi_3 -2
e_{\bmzero\bmone}{}^\alpha \partial_\alpha\phi_4 + 2
e_{\bmzero\bmzero}{}^\alpha \partial_\alpha \phi_3 \\ &&\hspace{1cm}=
(2\Gamma_{\bmzero\bmzero\bmone\bmone}-8
\Gamma_{\bmone\bmzero\bmone\bmzero})\phi_4 +
(4\Gamma_{\bmzero\bmzero\bmzero\bmone}+ 8
\Gamma_{\bmone\bmzero\bmzero\bmzero})\phi_3 - 6
\Gamma_{\bmzero\bmzero\bmzero\bmzero}\phi_2. \\
\end{eqnarray*}

The evolution system given by equations \eqref{p1}-\eqref{p7} and
\eqref{HR:BoundaryAdaptedSystem1}-\eqref{HR:BoundaryAdaptedSystem2}
has five characteristic modes, two of which propagate along null
directions transversal to the boundary, one ingoing and one outgoing,
with the other three timelike propagating and at least one remaining
tangent to the boundary. A computation shows that the characteristic
polynomial of the boundary adapted Bianchi system is given by
\begin{equation} 4 \big(\tau^\mu \xi_\mu\big) \big(
g^{\nu\lambda}\xi_\nu\xi_\lambda \big)\big(h^{\rho\sigma} \xi_\rho
\xi_\sigma \big) \nonumber
\end{equation} where
\begin{equation} h^{\rho\sigma}\equiv \tau^\rho \tau^\sigma +
\bme_{\bmzero\bmzero}{}^{(\rho}\bme_{\bmone\bmone}{}^{\sigma)}.
\nonumber
\end{equation} The vanishing of any one of the above factors via an
appropriate choice of $\xi_\mu$ tells us that $\xi_\mu$ is the
co-normal to a characteristic surface, along which information
propagates. By considering the first two factors, it is clear that
there is one timelike propagating mode and two null. The vanishing of
the final factor $h^{\mu\nu}\xi_\mu\xi_\nu$ yields the final two.

In the next section we take $\T$ to be given by a level set of the
coordinate $x^3$ and fix $\tau^\mu = \delta^\mu_0$. For these last two
modes to propagate along $\tau^{AA'}$, and hence tangential to the
boundary, one must enforce
\begin{equation}\label{e00e11zero} \bme_{\bmzero\bmzero}{}^3 = 0 =
\bme_{\bmone\bmone}{}^3.
\end{equation} This has the consequence that $l^{\rho\sigma}$ has no
component transverse to $\T$ (as $\tau^\mu = \delta^\mu_0$), so, for
example, $\xi_\mu = \pm\delta_\mu^3$ are co-normals to two further
characteristic surfaces within $\T$. Thus, for this choice of frame,
only the null modes propagate transverse to $\T$, a major
simplification.

There are two options to enforce this requirement. Generalising
equation \eqref{p2} by not enforcing parallel-propagatation with
respect to the Weyl connection, we can write the evolution equation
\begin{equation} \partial_t e_{\bmA\bmB}{}^3 =
-\chi_{(\bmA\bmB)}{}^{\bmP\bmQ}e_{\bmP\bmQ}{}^3 +
Z_{(\bmA\bmB)}{}^{\bmP\bmQ}e_{\bmP\bmQ}{}^3
\end{equation} where the extra term has enough degrees of freedom for
our needs. To satisfy equation \eqref{e00e11zero}, looking at the
$\bmzero\bmzero$ and $\bmone\bmone$ components, homogeneity shows that
if $\bme_{\bmzero\bmzero}{}^3 = \bme_{\bmone\bmone}{}^3 = 0$ initially
then they will remain zero if we choose
$Z_{\bmzero\bmzero}{}^{\bmP\bmQ} = \chi_{\bmzero\bmzero}{}^{\bmP\bmQ}$
and $Z_{\bmone\bmone}{}^{\bmP\bmQ} =
\chi_{\bmone\bmone}{}^{\bmP\bmQ}$. Hence, one option is to move away
from Weyl propagation. This option has a cascading effect whereby many
standard results computed assuming Weyl propagation no longer hold,
and is pursued elsewhere \cite{frauendiener2025fullynonlinear}. In
this work we will take an alternate approach, detailed further in
Sec.~\ref{Section:BoundaryGaugeConditions}, by keeping Weyl
propagation (i.e. $Z_{\bmA\bmB}{}^{\bmP\bmQ}=0$) and restricting the
geometry of our space-time to satisfy
$\chi_{\bmzero\bmzero}{}^{\bmzero\bmone} +
\chi_{\bmzero\bmzero}{}^{\bmone\bmzero} =
\chi_{\bmone\bmone}{}^{\bmzero\bmone} +
\chi_{\bmone\bmone}{}^{\bmone\bmzero} = 0$ as a means to enforce
equation \eqref{e00e11zero}.

The above discussion is in contrast to the standard evolution system
which has one zero-mode, an ingoing and an outgoing timelike mode and
an ingoing and an outgoing lightlike mode ---as can be read from the
characteristic polynomial of the subsystem which is given by
\begin{equation} 36\big(\tau^\mu\xi_\mu\big)\big(
g^{\nu\lambda}\xi_\nu\xi_\lambda \big) \left(\tau^\rho\tau^\sigma
+\frac{2}{3}g^{\rho\sigma}\right)\xi_\rho \xi_\sigma.  \nonumber
\end{equation}

\section{Boundary gauge
conditions}\label{Section:BoundaryGaugeConditions} In this section we
discuss the way the conformal Gau\ss ian gauge system discussed in
Section \ref{Subsection:HyperbolicReduction} can be implemented in the
context of an initial boundary value problem.

value problem for the evolution system consisting of the transport
equations \eqref{p1}-\eqref{p7} and the boundary adapted Bianchi
system
\eqref{HR:BoundaryAdaptedSystem1}-\eqref{HR:BoundaryAdaptedSystem2}. As
the hyperbolic reduction procedure leading to this system makes use of
a conformal Gaussian gauge, it is necessary to analyse to which extent
this gauge system is compatible with the existence of a timelike
boundary.

\subsection{A frame adapted to the
boundary}\label{Section:AdaptedFrame}

We construct a boundary adapted gauge akin to that used in
\cite{friedrich1999initial}. Let $\mathcal{M}$ be a smooth,
4-dimensional manifold $\mathcal{M} = \mathbb{R}^+_0 \times \Ss$ where
$\Ss$ is a 3-dimensional submanifold with a single, finite
2-dimensional boundary $\partial\mathcal{S}$. Let $\Sstar \simeq \{0\}
\times \Ss$, and $\T \simeq \mathbb{R}^+_0 \times
\partial\mathcal{S}$. Our central assumption is the following:

\begin{assumption}
\label{Assumption:BoundaryRuledByConformalGeodesics} The timelike
hypersurface $\mathcal{T}$ is ruled by a congruence of timelike
conformal geodesics. This congruences extends to a congruence of
timelike conformal geodesics covering a neighbourhood of $\mathcal{T}$
in $\mathcal{M}$.
\end{assumption}

This assumption aligns with the natural way of defining $\T$ as a
level set of a spatial coordinate. In the conformal Gau\ss ian gauge,
the spatial coordinates are constant along timelike coordinate
geodesics. The implications of the above assumption and the
restrictions it imposes on the initial boundary value problem will be
addressed shortly.

\smallskip In addition to the above assumption, and for simplicity of
the presentation it is further assumed that $\T$ is orthogonal to
$\Sstar$ on $\partial\mathcal{S}_\star\equiv
\{0\}\times\partial\mathcal{S}$. We introduce an orthonormal frame
$e_{\bm{a}}$, $\bm{a}=0,1,2,3$ on $\mathcal{M}$ with $\bm{e}_{\bm{0}}$
and $\bm{e}_{\bm{i}}$, $\bm{i}=0,1,2$ normal and intrinsic to
$\mathcal{S}_\star$ respectively, so that $e_{\bm{3}}$ is orthogonal
to $\T$ on $\partial\mathcal{S}_\star$. Crucially, we choose
$\bme_\bmzero$ (the time leg of the frame) to be tangent to the
conformal geodesics ruling $\mathcal{T}$. The remaining vectors of the
frame $\{ \bme_\bmi \}$ are Weyl propagated along the congruence of
conformal geodesics.

From the point of view of an initial value problem, the congruence of
conformal geodesics covering $\mathcal{M}$ in a neighbourhood of
$\mathcal{T}$ is specified through initial conditions for the curves
on $\mathcal{S}_\star$. The usual way of prescribing these initial
conditions is to require that
\begin{equation} \dot{x}^a_\star \perp \mathcal{S}_\star, \qquad
(\tilde{\beta}_a\dot{x}^a)_\star=0.
\label{CG:InitialData}
\end{equation} For consistency we require these conditions also at
$\partial\mathcal{S}_\star$. Observe that conditions
\eqref{CG:InitialData} do not fix the components of $\tilde{\beta}_a$
transverse to the congruence of conformal geodesics as $\dot{x}^a =
e^a_0$. The further specification of the spatial components of
$\tilde{\beta}_a$ will addressed later in this article. Given an
orthonormal spatial frame $\{ \bme_\bmi \}$ on $\mathcal{S}_\star$
such that $\bme_\bmthree$ is orthogonal to
$\partial\mathcal{S}_\star$, it is evolved off $\mathcal{S}_\star$ by
means of Weyl propagation $\tau^a\hat{\nabla}_ae^b_i = 0$. It is
important to stress that this requirement is consistent with the
evolution equation \eqref{FrameP1} ---or alternatively
\eqref{p1}-\eqref{p2}. In fact, it can be verified that the evolution
equation \eqref{FrameP1} together with the gauge conditions implied by
the conformal Gau\ss ian gauge imply the Weyl propagation of the
frame.

The construction of a conformal Gau\ss ian gauge adapted to the
timelike boundary $\mathcal{T}$ is supplemented by a prescription of
coordinates.  We choose coordinates $(x^\mu)$, $\mu=0,1,2,3$ so that
\begin{equation} \Sstar = \{x^0 = 0\}, \qquad \T = \{x^3 = 0\},
\nonumber
\end{equation} and propagate them into $\mathcal{M}$ through
\begin{equation} \bm{e}_{\bm{0}}(x^\mu) = \delta_0{}^\mu.  \nonumber
\end{equation}

The gauge conditions can then be summarised as
\begin{equation}\label{eq:gaussgaugeconds}
\text{$\bm{e}_{\bm{0}}{}^\mu = \delta_{\bm{0}}{}^\mu$ on $\mathcal{M}$
\qquad and \quad $e_{\bm{a}}{}^3 = e_{\bm{3}}{}^3 \delta_{\bm{a}}{}^3$
on $\T$.}
\end{equation}

$\T=\{x^3=0\}$. Define a frame $e_{\bm{i}}$ and let $e_{\bm{3}}$ be
the inward pointing space-like orthonormal vector to $\T$ such that
$\epsilon \equiv g(e_{\bm{3}},e_{\bm{3}}) = -1$. The remaining
$e_{\bm{i}}$, $i=0,1,2$, are parallely propagated along $e_{\bm{3}}$
so that $\nabla_{\bm{3}}e_{\bm{i}} = 0 =
\Gamma_{\bm{3}}{}^{\bm{j}}{}_{\bm{i}}e_{\bm{j}}$, yielding

so that $e_{\bm{3}}{}^\mu = e_{\bm{3}}{}^3\delta_{\bm{3}}{}^\mu$.

\subsection{Consequences of Assumption
\ref{Assumption:BoundaryRuledByConformalGeodesics}} After having
formulated the general geometric setting of the initial value problem
under consideration, in this subsection we analyse the restrictions
imposed by Assumption
\ref{Assumption:BoundaryRuledByConformalGeodesics}.

\medskip Setting $\dot{x}^\mu = z^{\bm{a}}e_{\bm{a}}{}^\mu$ and
$\bm{\beta} = \beta_{\bm{a}}\bm{\omega}^{\bm{a}}$, we can split the
conformal geodesic equations into parts normal and intrinsic to $\T$,
which yields ($\bm{a} = 0,1,2,3$ and $\bm{i} = 0,1,2$)

\begin{subequations}
\begin{align} \dot{x}^3 &= z^{\bm{a}}e_{\bm{a}}{}^3 =
z^{\bm{3}}, \label{eq:CGE3-1}\\ \dot{z}^{\bm{3}} &=
-\Gamma_{\bm{a}}{}^{\bm{3}}{}_{\bm{b}}z^{\bm{a}}z^{\bm{b}} -
2(\beta_{\bm{c}}z^{\bm{c}})z^{\bm{3}} +
(z_{\bm{c}}z^{\bm{c}})\beta^{\bm{3}}, \label{eq:CGE3-2}\\
\dot{\beta}_{\bm{3}} &=
\Gamma_{\bm{a}}{}^{\bm{c}}{}_{\bm{3}}z^{\bm{a}}\beta_{\bm{c}}
+(\beta_{\bm{c}}z^{\bm{c}})\beta_{\bm{3}}
-\frac12(\beta_{\bm{c}}\beta^{\bm{c}})z_{\bm{3}} +
L_{\bm{33}}z^{\bm{3}} + L_{\bm{i3}}z^{\bm{i}},
    \label{eq:CGE3-3}
\end{align}
\end{subequations} and
\begin{subequations}
\begin{align} \dot{x}^\alpha &= e_{\bm{a}}{}^\alpha
z^{\bm{a}}, \label{eq:CGEa-1}\\ \dot{z}^{\bm{i}} &=
-\Gamma_{\bm{a}}{}^{\bm{i}}{}_{\bm{b}}z^{\bm{a}}z^{\bm{b}} -
2(\beta_{\bm{c}}z^{\bm{c}})z^{\bm{i}} +
(z_{\bm{c}}z^{\bm{c}})\beta^{\bm{i}}, \label{eq:CGEa-2}\\
\dot{\beta}_{\bm{i}} &=
\Gamma_{\bm{b}}{}^{\bm{c}}{}_{\bm{i}}z^{\bm{b}}\beta_{\bm{c}}
+(\beta_{\bm{c}}z^{\bm{c}})\beta_{\bm{i}}
-\frac12(\beta_{\bm{c}}\beta^{\bm{c}})z_{\bm{i}} +
L_{\bm{3i}}z^{\bm{3}} + L_{\bm{ci}}z^{\bm{c}}.\label{eq:CGEa-3}
\end{align}
\end{subequations}

The right-hand side of the normal subsystem, equations
\eqref{eq:CGE3-1}-\eqref{eq:CGE3-3}, vanishes if we initially have the
conformal geodesic intrinsic to $\T$ so that $z^{\bm{3}} =
\beta_{\bm{3}} = 0$ and we impose
\begin{equation}\label{eq:CGConds1} L_{\bm{i3}} = 0, \qquad
\Gamma_{\bm{i}}{}^{\bm{3}}{}_{\bm{j}} = 0, \qquad
\Gamma_{\bm{i}}{}^{\bm{j}}{}_{\bm{3}} = 0,\qquad i=0,1,2.
\end{equation} If the above hold, then the normal subsystem takes the
form
\begin{subequations}
\begin{align} \dot{x}^3 &= z^{\bm{3}}, \label{eq:CGENormSubSys-1}\\
\dot{z}^{\bm{3}} &= - 2(\beta_{\bm{c}}z^{\bm{c}})z^{\bm{3}} +
(z_{\bm{c}}z^{\bm{c}})\beta^{\bm{3}}, \label{eq:CGENormSubSys-2}\\
\dot{\beta}_{\bm{3}} &= (\beta_{\bm{3}}z^{\bm{3}})\beta_{\bm{3}}
-\frac12(\beta_{\bm{c}}\beta^{\bm{c}})z_{\bm{3}} +
L_{\bm{33}}z^{\bm{3}}.
    \label{eq:CGENormSubSys-3}
\end{align}
\end{subequations} Homogeneity implies that if $\dot{x}^3 =
\dot{z}^{\bm{3}} = \dot{\beta}_{\bm{3}} = 0$ initially, then the only
solution of equations
\eqref{eq:CGENormSubSys-1}-\eqref{eq:CGENormSubSys-3} is the vanishing
one. In turn, the intrinsic equations
\eqref{eq:CGEa-1}-\eqref{eq:CGEa-3} reduce to
\begin{align*} \dot{x}^\alpha &= e_{\bm{i}}{}^\alpha z^{\bm{i}}, \\
\dot{z}^{\bm{i}} &=
-\Gamma_{\bm{j}}{}^{\bm{i}}{}_{\bm{k}}z^{\bm{j}}z^{\bm{k}} -
2(\beta_{\bm{k}}z^{\bm{k}})z^{\bm{i}} +
(z_{\bm{k}}z^{\bm{k}})\beta^{\bm{i}}, \\ \dot{\beta}_{\bm{i}} &=
\Gamma_{\bm{j}}{}^{\bm{k}}{}_{\bm{i}}z^{\bm{j}}\beta_{\bm{k}}
+(\beta_{\bm{k}}z^{\bm{k}})\beta_{\bm{i}}
-\frac12(\beta_{\bm{k}}\beta^{\bm{k}})z_{\bm{i}} +
L_{\bm{ki}}z^{\bm{k}}.
\end{align*} The equations coincide with the conformal geodesic
equations with respect to the induced metric $\bm{h}$ on $\mathcal{T}$
if and only if the part of the 4-Schouten tensor intrinsic to
$\mathcal{T}$ equals the 3-Schouten tensor intrinsic to $\T$. That is,
if one has that
\begin{equation}\label{eq:SchoutenCond} L_{\bm{ij}} = l_{\bm{ij}},
\end{equation} where $l_{\bmi\bmj}$ denotes the components of the
Schouten tensor of the intrinsic 3-metric $\bm{h}$ of $\mathcal{T}$.

The latter is one of the central observations of our analysis. We
summarise the above as follows:

\begin{lemma}
\label{Lemma:IntrinsicCG} The conformal Gau\ss ian system discussed in
Subsection \ref{Subsection:HyperbolicReduction} with the adapted frame
of Subsection \ref{Section:AdaptedFrame} implies a congruence of
conformal geodesics intrinsic to $\mathcal{T}$ if and only if
Conditions \eqref{eq:CGConds1} and \eqref{eq:SchoutenCond} hold.
\end{lemma}

\subsubsection{Consistency of the gauge} We now proceed to verify the
consistency of the conditions in Lemma \ref{Lemma:IntrinsicCG} with
our gauge. To this end we have the following lemmas:

\begin{lemma} The condition $\Gamma_{\bm{i}}{}^{\bm{3}}{}_{\bm{j}} =
0$ is satisfied if $e_{\bm{i}}{}^3=0$.
\end{lemma}

\begin{proof} The frame coefficients satisfy $e_{\bm{i}}{}^3=0$ ---see
equation ~\eqref{eq:gaussgaugeconds}. Then one trivially has
\begin{equation} \bm{e}_{\bm{i}}(e_{\bm{j}}{}^3) -
\bm{e}_{\bm{j}}(e_{\bm{i}}{}^3)=0 \nonumber
\end{equation} and given that our connection is torsion-free, the Lie
bracket of the frame vectors intrinsic to $\T$ can be written as
\begin{align} [\bm{e}_{\bm{i}},\bm{e}_{\bm{j}}] &=
\big{(}\bm{e}_{\bm{i}}(e_{\bm{j}}{}^\alpha) -
\bm{e}_{\bm{j}}(e_{\bm{i}}{}^\alpha)\big{)}\bm{\partial}_\alpha,
\qquad \alpha=0,1,2 \nonumber\\ &=
\big{(}\bm{e}_{\bm{i}}(e_{\bm{j}}{}^\mu) -
\bm{e}_{\bm{j}}(e_{\bm{i}}{}^\mu)\big{)}\bm{\partial}_\mu, \qquad
\mu=0,1,2,3 \nonumber\\ &=
\big{(}\Gamma_{\bm{i}}{}^{\bm{c}}{}_{\bm{j}} -
\Gamma_{\bm{j}}{}^{\bm{c}}{}_{\bm{i}}\big{)}\bm{e}_{\bm{c}}, \nonumber
\\ &= \big{(}\Gamma_{\bm{i}}{}^{\bm{c}}{}_{\bm{j}} -
\Gamma_{\bm{j}}{}^{\bm{c}}{}_{\bm{i}}\big{)}e_{\bm{c}}{}^{\mu}\bm{\partial}_\mu. \nonumber
\end{align} So, for consistency we must have
\begin{equation} \big{(}\Gamma_{\bm{i}}{}^{\bm{c}}{}_{\bm{j}} -
\Gamma_{\bm{j}}{}^{\bm{c}}{}_{\bm{i}}\big{)}e_{\bm{c}}{}^3 = 0,
\nonumber
\end{equation} which after using $e_{\bm{i}}{}^3=0$ and
$e_{\bm{3}}{}^3 \neq 0$ implies that
\begin{equation} \Gamma_{\bm{i}}{}^{\bm{3}}{}_{\bm{j}} -
\Gamma_{\bm{j}}{}^{\bm{3}}{}_{\bm{i}} = 0.  \nonumber
\end{equation} Together with metric compatibility
\begin{equation} \Gamma_{\bm{i}}{}^{\bm{k}}{}_{\bm{j}}\eta_{\bm{kl}} +
\Gamma_{\bm{i}}{}^{\bm{k}}{}_{\bm{l}}\eta_{\bm{kj}} = 0 \nonumber
\end{equation} we then obtain $\Gamma_{\bm{i}}{}^{\bm{3}}{}_{\bm{j}} =
0$ as required. So this restriction is a natural consequence of the
gauge.

\end{proof}

\begin{lemma} The condition $\Gamma_{\bm{i}}{}^{\bm{j}}{}_{\bm{3}} =
0$ is satisfied if and only if
$e_{\bm{a}}{}^3=e_{\bm{3}}{}^3\delta_{\bm{a}}{}^3$.
\end{lemma}

\begin{proof} Our gauge satisfies $e_{\bm{a}}{}^3 = e_{\bm{3}}{}^3
\delta_{\bm{a}}{}^3$ (see equation~\ref{eq:gaussgaugeconds}) which
implies $e_{\bm{i}}(e_{\bm{3}}{}^\alpha)\bm{\partial}_\alpha = 0$ for
$\alpha = 0,1,2$ and immediately yields
$\Gamma_{\bm{i}}{}^{\bm{j}}{}_{\bm{3}} = 0$ as required. Proceeding
conversely, suppose $\Gamma_{\bm{i}}{}^{\bm{j}}{}_{\bm{3}} = 0$. Then
$e_{\bm{i}}(e_{\bm{3}}{}^\mu) = 0$ and so we must have
$e_{\bm{3}}{}^\mu = e_{\bm{3}}{}^3\delta_3{}^\mu + c^\mu$ with
constants $c^\mu$. If $e_{\bm{3}}$ is initially intrinsic to $\T$,
which is true for our gauge choice, then $c^\mu=0$ and we recover
$e_{\bm{3}}{}^\mu = e_{\bm{3}}{}^3\delta_3{}^\mu$ as required.
\end{proof}

Thus, we conclude that the conditions on the connection in
\eqref{eq:CGConds1} are implied by our gauge conditions. The analysis
of the restrictions implied by the conditions on the components of the
Schouten tensor is more delicate and requires the use of the conformal
Einstein constraint equations. This is done in the next section.

\section{The conformal constraints}\label{sec:conformalconstraints}
\label{Section:ConformalConstraints} In this section we analyse the
implications of the conditions
\begin{equation} L_{\bmi\bmthree}=0, \qquad L_{\bmi\bmj}=l_{\bmi\bmj}
\quad \mbox{on}\quad \mathcal{T}
\label{CurvatureConditions}
\end{equation} ---cfr. equations \eqref{eq:CGConds1} and
\eqref{eq:SchoutenCond}. As these conditions involve the curvature, it
is necessary to consider the \emph{conformal Einstein constraint
equations} ---that is, the constraints implied by the conformal
Einstein equations on a hypersurface.

\subsection{Definitions and the equations}
 
 The conformal vacuum field equations intrinsic to a time-like surface
are \cite{kroon2017conformal}:
\begin{subequations}
\begin{gather} D_{\bm{i}}D_{\bm{j}}\Omega = \Sigma K_{\bm{ij}} -
\Omega L_{\bm{ij}} + sh_{\bm{ij}}, \label{eq:CCE-a} \\[4pt]
D_{\bm{i}}\Sigma = K_{\bm{i}}{}^{\bm{k}}D_{\bm{k}}\Omega - \Omega
L_{\bm{i}}, \label{eq:CCE-b} \\[4pt] D_{\bm{i}}s = \Sigma L_{\bm{i}} -
L_{\bm{ik}}D^{\bm{k}}\Omega, \label{eq:CCE-c} \\[4pt]
D_{\bm{i}}L_{\bm{jk}} -D_{\bm{j}}L_{\bm{ik}} = \Sigma d_{\bm{kij}} +
D^{\bm{l}}\Omega d_{\bm{lkij}} + K_{\bm{ik}}L_{\bm{j}} -
K_{\bm{jk}}L_{\bm{i}}, \label{eq:CCE-d} \\[4pt] D_{\bm{i}}L_{\bm{j}} -
D_{\bm{j}}L_{\bm{i}} = D^{\bm{l}}\Omega d_{\bm{lij}} +
K_{\bm{i}}{}^{\bm{k}}L_{\bm{jk}} -
K_{\bm{j}}{}^{\bm{k}}L_{\bm{ik}}, \label{eq:CCE-e} \\[4pt]
D^{\bm{k}}d_{\bm{kij}} = K^{\bm{k}}{}_{\bm{j}}d_{\bm{ik}} -
K^{\bm{k}}{}_{\bm{i}}d_{\bm{jk}}, \label{eq:CCE-f} \\[4pt]
D^{\bm{i}}d_{\bm{ij}} = K^{\bm{ik}}d_{\bm{ijk}}, \label{eq:CCE-g}
\\[4pt] \lambda = 6\Omega s + 3\Sigma^2 - 3D_{\bm{k}}\Omega
D^{\bm{k}}\Omega, \label{eq:CCE-h} \\[4pt] D_{\bm{j}}K_{\bm{ki}} -
D_{\bm{k}}K_{\bm{ji}} = \Omega d_{\bm{ijk}} + h_{\bm{ij}}L_{\bm{k}} -
h_{\bm{ik}}L_{\bm{j}}, \label{eq:CCE-i} \\[4pt] l_{\bm{ij}} = \Omega
d_{\bm{ij}} + L_{\bm{ij}} - K(K_{\bm{ij}} - \frac14Kh_{\bm{ij}}) +
K_{\bm{ki}}K_{\bm{j}}{}^{\bm{k}} -
\frac14K_{\bm{kl}}K^{\bm{kl}}h_{\bm{ij}}, \label{eq:CCE-j}
\end{gather}
\end{subequations} where $D_\bmi$ denotes the covariant derivative of
the metric $\bm{h}$ on $\mathcal{T}$, $\Omega$ denotes the restriction
of the spacetime conformal factor $\Theta$ to $\mathcal{T}$ while
$\Sigma$ encodes the normal derivative of $\Theta$. The field $K_{ij}$
is the extrinsic curvature of the hypersurface. Moreover, $L_\bmi$ and
$L_{\bmi\bmj}$ denote, respectively the normal-transverse and
transverse-transverse components of the 4-dimensional Schouten tensor
$L_{ab}$. The scalar $s$ is the so-called \emph{Friedrich scalar} ---a
combination of the D'Alembertian of the spacetime conformal factor and
the spacetime Ricci scalar. Finally, the restriction of the rescaled
Weyl tensor $d^a{}_{bcd}$ is encoded in the fields
\begin{equation} d_{\bmi\bmj}, \qquad d_{\bmi\bmj\bmk}, \qquad
d_{\bmi\bmj\bmk\bml}.  \nonumber
\end{equation} The fields $d_{\bmi\bmj}$ and $d_{\bmi\bmj\bmk}$
correspond, respectively, to the electric and magnetic parts of
$d^a{}_{bcd}$ with respect to the normal to $\mathcal{T}$. The
magnetic part $d_{\bmi\bmj\bmk}$ can be re-expressed in terms of a
rank 2 tensor $d^*_{\bmi\bmj}$ which is the 3-dimensional Hodge dual
of the former. The rank 4 field $d_{\bmi\bmj\bmk\bml}$ can be, in
turn, completely expressed in terms of the electric part
$d_{ij}$. \footnote{Strictly speaking, equations \eqref{eq:CCE-i} and
\eqref{eq:CCE-j} are not part of the conformal constraint equations
but rather the Codazzi-Mainardi and Gauss-Codazzi equations written in
terms of variables of the conformal equations.}

\medskip A more detailed discussion of equations
\eqref{eq:CCE-a}-\eqref{eq:CCE-j}, including interdependencies can be
found in Section 11.4 of \cite{kroon2017conformal}.

\subsection{Restrictions on the curvature} In terms of the notation
introduced in the previous subsection, conditions
\eqref{CurvatureConditions} can be rewritten as
\begin{equation} L_{\bmi}=0, \qquad L_{\bmi\bmj}=l_{\bmi\bmj} \quad
\mbox{on}\quad \mathcal{T}.
\end{equation} Substituting the latter into the conformal constraints
\eqref{eq:CCE-a}-\eqref{eq:CCE-j} one obtain the reduced expressions:
\begin{subequations}\label{eq:CCE}
\begin{gather} D_{\bm{i}}D_{\bm{j}}\Omega = \Sigma K_{\bm{ij}} -
\Omega l_{\bm{ij}} + sh_{\bm{ij}}, \label{Reducedeq:CCE-a} \\[4pt]
D_{\bm{i}}\Sigma =
K_{\bm{i}}{}^{\bm{k}}D_{\bm{k}}\Omega, \label{Reducedeq:CCE-b} \\[4pt]
D_{\bm{i}}s = - l_{\bm{ik}}D^{\bm{k}}\Omega, \label{Reducedeq:CCE-c}
\\[4pt] D_{\bm{i}}l_{\bm{jk}} -D_{\bm{j}}l_{\bm{ik}} = \Sigma
d_{\bm{kij}} + D^{\bm{l}}\Omega d_{\bm{lkij}}, \label{Reducedeq:CCE-d}
\\[4pt] D^{\bm{l}}\Omega d_{\bm{lij}} =
K_{\bm{j}}{}^{\bm{k}}l_{\bm{ik}} -
K_{\bm{i}}{}^{\bm{k}}l_{\bm{jk}}, \label{Reducedeq:CCE-e} \\[4pt]
D^{\bm{k}}d_{\bm{kij}} = K^{\bm{k}}{}_{\bm{j}}d_{\bm{ik}} -
K^{\bm{k}}{}_{\bm{i}}d_{\bm{jk}}, \label{Reducedeq:CCE-f} \\[4pt]
D^{\bm{i}}d_{\bm{ij}} =
K^{\bm{ik}}d_{\bm{ijk}}, \label{Reducedeq:CCE-g} \\[4pt] \lambda =
6\Omega s + 3\Sigma^2 - 3D_{\bm{k}}\Omega
D^{\bm{k}}\Omega, \label{Reducedeq:CCE-h} \\[4pt]
D_{\bm{j}}K_{\bm{ki}} - D_{\bm{k}}K_{\bm{ji}} = \Omega
d_{\bm{ijk}}, \label{Reducedeq:CCE-i} \\[4pt] \Omega d_{\bm{ij}}=
K(K_{\bm{ij}} - \frac14Kh_{\bm{ij}}) -
K_{\bm{ki}}K_{\bm{j}}{}^{\bm{k}} +
\frac14K_{\bm{kl}}K^{\bm{kl}}h_{\bm{ij}}. \label{Reducedeq:CCE-j}
\end{gather}
\end{subequations}

The case where $\mathcal{T}$ coincides with the conformal boundary so
that $\Omega = D_\bmi \Omega =0$ has been analysed in detail in
\cite{friedrich1995einstein} ---see also Section 11.4.4 in
\cite{kroon2017conformal}.  In this case the vanishing of the
conformal factor simplifies the conformal constraints in a substantial
manner allowing to solve for almost all unknowns save for
$d_{\bmi\bmj}$. This should be contrasted with the situation
considered here.

Of particular relevance for our analysis is the reduced Gauss-Codazzi
relation, equation \eqref{Reducedeq:CCE-j}, which expresses the
electric part of the rescaled Weyl tensor in terms of the conformal
factor $\Omega$ and the extrinsic curvature $K_{\bmi\bmj}$. Thus, the
electric part of the Weyl tensor cannot be freely specified if the
extrinsic curvature is already prescribed. As this observation is
essential for our analysis we state it in the form of a Lemma:

\begin{lemma} The following two conditions are equivalent:
\begin{itemize}
\item[(i)] $L_{\bmi\bmj}=l_{\bmi\bmj}$;
\item[(ii)] $ \Omega d_{\bm{ij}}= K(K_{\bm{ij}} - \frac14Kh_{\bm{ij}})
- K_{\bm{ki}}K_{\bm{j}}{}^{\bm{k}} +
\frac14K_{\bm{kl}}K^{\bm{kl}}h_{\bm{ij}}$.
\end{itemize}
\end{lemma}

A potentially interesting case to consider is that of an umbilical
timelike boundary $\T$ ---that is when the extrinsic curvature is pure
trace so that
\begin{equation} K_{\bm{ij}} = \frac{1}{4}Kh_{\bm{ij}}.  \nonumber
\end{equation} In this case condition \eqref{Reducedeq:CCE-j} readily
gives that
\begin{equation} 64\Omega d_{\bmi\bmj} = K^2 h_{\bmi\bmj}.  \nonumber
\end{equation} Thus, given that $d_{ij}$ is a trace-free tensor one
readily concludes that $K=0$ and, moreover that $K_{\bmi\bmj}=0$. Now,
using the conformal Codazzi-Mainardi equation \eqref{Reducedeq:CCE-i}
assuming that $L_i=0$ one readily concludes, in turn, that
\begin{equation} d_{\bmi\bmj\bmk}=0.  \nonumber
\end{equation} Thus, in the case of an umbilical timelike boundary the
whole of the rescaled Weyl tensor vanishes ---as it will be seen in
Section \ref{Section:Formulation}, under these circumstances no
non-trivial maximally dissipative boundary conditions can be imposed
on $\mathcal{T}$. This example illustrates the delicate interplay
between the conformal constraints, the gauge and the boundary
conditions.

\medskip A final remark concerns the Cotton tensor. The Cotton tensor
of the metric ${\bm l}$, $y_{ijk}$, is defined through the relation
\begin{equation} y_{ijk}= D_i l_{jk}- D_{j}l_{ik}.  \nonumber
\end{equation} It is a conformal invariant and, thus, it encodes
information about the conformal class $[{\bm l}]$. In particular, if
the metric ${\bm l}$ is conformally flat then $y_{ijk}=0$. In terms of
the Cotton tensor, the reduced conformal constraint
\eqref{Reducedeq:CCE-d} can be rewritten as
\begin{equation}\label{eq:cotton} y_{\bmi\bmj\bmk} = \Sigma
d_{\bmk\bmi\bmj} + D^\bml\Omega d_{\bml\bmk\bmi\bmj}.  \nonumber
\end{equation} Accordingly, the Cotton tensor is related to the
electric part of the rescaled Weyl tensor. The latter, in turn, via
the Codazzi-Mainardi identity, equation \eqref{Reducedeq:CCE-i}, is
related to the anti-symmetrised derivative of the extrinsic curvature.

\subsection{The Gauss constraints on $\mathcal{T}$} The Gauss
constraints \eqref{Reducedeq:CCE-f} and \eqref{Reducedeq:CCE-g} are
central to establish the connection between solutions to the conformal
constraints and any boundary condition we consider prescribing on the
boundary $\mathcal{T}$.

\medskip Equations \eqref{Reducedeq:CCE-f}-\eqref{Reducedeq:CCE-g} can
be brought into a more symmetric form by introducing the symmetric
trace-free tensor
\begin{equation} d^*_{ij}\equiv
-\frac{1}{2}\epsilon_j{}^{kl}d_{ikl}. \nonumber
\end{equation} This tensor encodes the same information as
$d_{ijk}$. Making use of this definition one obtains the following
alternative form of the Gauss constraints:
\begin{subequations}
\begin{eqnarray} && D^\bmi d^*_{\bmi\bmj} = \epsilon_\bmj{}^{\bmk\bml}
K^\bmi{}_\bmk d_{\bml\bmi}, \label{GaussTensorial1}\\ && D^\bmi
d_{\bmi\bmj} =
\epsilon^\bml{}_{\bmj\bmk}K^{\bmi\bmk}d_{\bmi\bml}^*. \label{GaussTensorial2}
\end{eqnarray}
\end{subequations} In a Riemannian context (for a fixed geometry)
equations \eqref{GaussTensorial1}-\eqref{GaussTensorial2} constitute a
system of underdetermined elliptic equations. In a 1+2 Lorentzian
setting, however, these equations can be cast as an evolution system
for certain components of $d_{\bmi\bmj}$ and $d^*_{\bmi\bmj}$ if some
of the other components are known ---in which case the latter act as a
source. The associated \emph{hyperbolic reduction procedure} is best
discussed making use of a 1+1+2 spinor decomposition.

\subsection{A 1+1+2 spinor decomposition of the Gauss
constraints}\label{Section:Decomposition:GaussConstraints}

As pointed out in, for example \cite{kroon2017conformal}, the
space-spinor formalism \cite{sommers1980space} provides a powerful
tool to carry out hyperblic reductions of geometric differential
systems like the conformal Einstein field equations. The key
observation is that the required split between evolution and
constraint equations follows from the irreducible decomposition of the
spinorial equations.

The usual space-spinor decomposition (also called the
$SU(2)$-decomposition) as first introduced in \cite{sommers1980space}
is based on the existence of a timelike vector field $\tau^a$ which
can be used as a map from the space of primed spinors
$\mathfrak{S}_{A'}$ to the space of unprimed spinors $\mathfrak{S}_A$
---see the discussion in Subsection
\ref{Subsubsection:SpaceSpinorSplit}. As already mentioned, one of the
attractions of this formalism is that it yields evolution equations
along $\tau^a$ and constraint equations orthogonal to $\tau^a$ by way
of irreducible decompositions of the resulting unprimed
spinors. Totally symmetric spinors contain information orthogonal to
$\tau^a$ while the anti-symmetric pieces contain information along
$\tau^a$. Such a formalism can also be introduced adapted to a
spacelike vector field $\rho^a$ in an analogous fashion. An important
difference now being that the resulting three-dimensional spaces
orthogonal to $\rho^a$ are \emph{Lorentzian}. This implies that the
conformal constraints intrinsic to one of these surfaces can be
further split into a \emph{reduced evolution and constraint system}
intrinsic to the surface. Here we outline the main details toward a
1+1+2 spinor splitting, referring to \ref{app:spinorformalism} for a
more rigorous and complete exposition.

\subsubsection{Decomposition with respect to $\rho^a$.} In the
following consider a spacelike vector field $\rho^a$ with
normalisation $\rho^a\rho_a = -2$. Further, it will be assumed that
one has an \emph{adapted spinor dyad} $\{o^A, \iota^A\}$ with
$o_A\iota^A=1$ such that the spinor counterpart $\rho^{AA'}$ takes the
form
\begin{equation} \rho^{AA'} = o^Ao^{A'} - \iota^A\iota^{A'}.
\label{nonumber}
\end{equation} The obvious map $ \mathfrak{S}_{A'} \rightarrow
\mathfrak{S}_A$ induced by $\rho^{AA'}$ and defined by $T_A \equiv
\rho_A{}^{A'}T_{A'}$ is not suitable to define an isomorphism between
primed and unprimed indices as it does not preserve norms ---this
easily follows by observing that if one defines $\rho_{AB} \equiv
\rho_A{}^{A'}\rho_{BA'} = -\epsilon_{AB}$ then $\rho_{AB}\rho^{AB} =
2$. We instead use the norm-preserving map $ \mathfrak{S}_{A'}
\rightarrow \mathfrak{S}_A$ defined by $T_A \equiv
\mbox{i}\rho_A{}^{A'}T_{A'}$. We call the resulting spinors
$SU(1,1)$-spinors. For convenience, in the following we write
\begin{equation} G^{AA'} \equiv \mbox{i} \rho^{AA'}.  \nonumber
\end{equation}

\smallskip The formalism induced by the spinor $G^{AA'}$ gives rise to
the following split of the covariant derivative operator
$\nabla_{AA'}$:
\begin{equation} G_B{}^{A'}\nabla_{AA'}\mu_C = D_{AB}\mu_C +
\frac12\epsilon_{AB}\TD\mu_C, \nonumber
\end{equation} where $\TD_{AB} \equiv G_{(A}{}^{A'}\nabla_{B)A'}$ is
the \emph{Sen connection} and $\TD \equiv G^{AA'}\nabla_{AA'}$ is the
\emph{Fermi connection}.

\smallskip Now, suppose that one has performed the mapping induced by
$G^{AA'}$ on a general spinor, resulting in an expression containing
only unprimed indices. The irreducible decomposition can then be
computed, which decomposes the expression into all the various
combinations of symmetrisation and anti-symmetrisation. The
totally-symmetric pieces contain information orthogonal to $\rho^a$,
intrinsic to some timelike hypersurface $\T$, while the antisymmetric
pieces, which can be represented as a trace multiplied by an
$\epsilon_{AB}$, contain information along $\rho^a$.

\subsubsection{Decomposition with respect to $\tau^a$.} Given a
timelike vector field
\begin{equation} \tau^{AA'} = o^Ao^{A'} + \iota^A\iota^{A'} \nonumber
\end{equation} intrinsic to a timelike surface $\T$, we can perform a
further split. We only consider totally-symmetric spinors, given that
this is the form of the spinors intrinsic to $\T$ after the mapping
with respect to $G^{AA'}$ and computing the irreducible
decomposition. We define $\tau_{AB} \equiv G_B{}^{A'}\tau_{AA'}$. In
this spirit, the Sen connection $D_{AB}$ can be decomposed into two
derivative operators intrinsic to $\T$ via the relation
\begin{equation} D_{AB} = \partial_{AB} + \frac12\tau_{AB}\partial,
    \label{DecompositionD}
\end{equation} where
\begin{equation} \tau^{AB}\partial_{AB} = 0.  \nonumber
\end{equation}

\subsubsection{The 1+1+2 evolution equations.}
\label{Section:DecompositionBianchiEqns} Now, starting from the spin-2
zero rest-mass equation
\begin{equation}\label{eq:spin2} \nabla^A{}_{A'}\phi_{ABCD} = 0
\end{equation} governing the Weyl spinor $\phi_{ABCD}=\phi_{(ABCD)}$,
it readily follows that its $SU(1,1)$-spinor decomposition yields
\begin{equation} D_{DE}\phi_{ABC}{}^E + \frac12 D\phi_{ABCD} = 0.
\nonumber
\end{equation} This has the irreducible decomposition
\begin{equation} D_{CD}\phi_{AB}{}^{CD} = 0, \qquad
D_{E(C}\phi_{D)AB}{}^E + \frac12 D\phi_{ABCD} = 0 \nonumber
\end{equation} in the usual five ``evolution equations'' and three
``constraint equations'' form ---cfr. the analogous decomposition in
the $SU(2)$ case in Section \ref{Subsubsection:SpaceSpinorSplit}. In
order to obtain further insights, we now consider the \emph{electric
and magnetic parts of the rescaled Weyl spinor} defined through its
tensor representation $d_{abcd}$ as
\begin{eqnarray*} && E_{ABCD} \equiv
\frac12\rho_B{}^{A'}\rho^{EE'}\rho_D{}^{C'}\rho^{FF'}d_{AA'EE'CC'FF'},
\\ && B_{ABCD} \equiv
\frac12\rho_B{}^{A'}\rho^{EE'}\rho_D{}^{C'}\rho^{FF'}d^*_{AA'EE'CC'FF'},
\\
\end{eqnarray*} where $d^*_{AA'BB'CC'DD'}$ is the spinor counterpart
of the right-dual of $d_{ij}$ defined in the usual way. The electric
and magnetic parts are related to to the rescaled Weyl spinor through
\begin{gather*} E_{ABCD} = -\frac12(\phi_{ABCD} +
\phi_{ABCD}^\ddagger), \qquad B_{ABCD} = -\frac{i}{2}(\phi_{ABCD} -
\phi_{ABCD}^\ddagger), \\ \phi_{ABCD} = -E_{ABCD} + iB_{ABCD}.
\end{gather*} Using these expressions together with the equations
intrinsic to $\T$, and using that $E_{ABCD}$ and $B_{ABCD}$ are real
spinors, yield the individual equations
\begin{equation} D^{CD}E_{ABCD} = 0, \qquad D^{CD}B_{ABCD} =
0. \nonumber
\end{equation} Now, as discussed in \ref{app:spinorformalism}, the
spinor $B_{ABCD}$ admits the decomposition
\begin{equation} B_{ABCD} = \mu_{ABCD} + \tau_{AB}\mu_{CD} +
\mu_{AB}\tau_{CD} - \mu\big{(} 3\tau_{AB}\tau_{CD} +
2\epsilon_{C(A}\epsilon_{B)D} \big{)} \nonumber
\end{equation} with
\begin{equation} \mu_{ABCD} = \mu_{(ABCD)}, \quad \mu_{AB} =
\mu_{(AB)}, \quad \mu_{ABCD}\tau^{AB} = 0, \quad \mu_{AB}\tau^{AB} = 0
\nonumber
\end{equation} ---see equation \eqref{DecompositionValence4}--- so
that combining with the decomposition \eqref{DecompositionD} of the
derivative operator $D_{AB}$ one obtains the following two evolution
equations for the components of $B_{ABCD}$:
\begin{subequations}
\begin{align} 4\partial\mu - 2\partial^{AB}\mu_{AB} &=
-2\mu_{AB}\partial\tau^{AB} - 6\mu\partial^{AB}\tau_{AB} +
\tau^{CD}\partial^{AB}\mu_{ABCD}, \label{BianchiTransport1}\\
2\partial\mu_{AB} + 4\partial_{AB}\mu &=
-2(\mu_{AB}\partial^{CD}\tau_{CD} + \mu_{CD}\partial^{CD}\tau_{AB}) +
6\mu\partial\tau_{AB} \nonumber \\ &\quad+ (\mu_{AB}{}^{CD} -
\tau_{AB}\mu^{CD})\partial\tau_{CD} + (\tau_{AB}\tau^{EF} -
2\delta_A{}^E\delta_B{}^F)\partial^{CD}\mu_{CDEF}. \label{BianchiTransport2}
\end{align}
\end{subequations} The above equations imply a symmetric hyperbolic
evolution system for the components $\mu$ and $\mu_{AB}$ provided that
the valence 4 component $\mu_{ABCD}$ is known. These equations will be
essential in the subsequent analysis.

\medskip Note, a similar set of equation can be obtained,
\emph{mutatis mutandi}, for the components of the spinor
$E_{ABCD}$. For reasons to be clarified in the sequel, these equations
will not be needed in our analysis.

\section{The formulation of the initial boundary value
problem}\label{Section:Formulation} Having discussed all the required
building blocks for our construction, in this section we discuss the
formulation of an initial boundary value problem for the extended
conformal Einstein equations with boundary data prescribed on a
timelike hypersurface $\mathcal{T}$ and initial data on a spacelike
hypersurface $\mathcal{S}_\star$.

\medskip The analysis of the preceding sections show that, in terms of
the conformal Gau\ss ian gauge adapted to the timelike boundary
$\mathcal{T}$, the extended conformal Einstein equations imply an
evolution system which, schematically, takes the form
\begin{subequations}
\begin{eqnarray} && \partial_\tau \bmupsilon =\mathbf{K} \bmupsilon +
\mathbf{Q}(\bmGamma) \bmupsilon + \mathbf{L}(x)
\bmphi, \label{SchematicEvolutionSystem1} \\ && (\mathbf{I} +
\mathbf{A}^0(\bme))\partial_\tau \bmphi + \mathbf{A}^\alpha (\bme)
\partial_\alpha \bmphi = \mathbf{B}(\bmGamma) \bmphi
\label{SchematicEvolutionSystem2}
\end{eqnarray}
\end{subequations} where the vector-valued unknown $\bmupsilon$
contains the independent components of the frame, connection and
Schouten tensor while $\bmphi$ encodes the independent components of
the rescaled Weyl spinor. Moreover, $\mathbf{I}$ is the $5\times 5$
unit matrix, $\mathbf{K}$ is a constant matrix, $\mathbf{Q}(\bmGamma)$
is a smooth matrix-valued function depending on the components of the
connection, $\mathbf{L}(x)$ is a smooth matrix-valued function of the
coordinates, $\mathbf{A}^\mu(\bme)$ are Hermitian matrices depending
smoothly on the frame coefficients and, finally,
$\mathbf{B}(\bmGamma)$ is a smooth matrix-valued function of the
connection coefficients.  For further details and discussion see
Proposition 13.3 in \cite{kroon2017conformal}.

For future reference it is noticed that the principal part of the
Bianchi subsystem \eqref{SchematicEvolutionSystem2} takes the form
\begin{equation} \left(
\begin{array}{ccccc} \delta_0^\mu + 2 e_{\bmzero\bmone}{}^\mu & -
2e_{\bmzero\bmzero}{}^\mu & 0 & 0 & 0 \\ 2e_{\bmone\bmone}{}^\mu &
2\delta_0{}^\mu & - 2e_{\bmzero\bmzero}{}^\mu & 0 & 0\\ 0 & 2
e_{\bmone\bmone}{}^\mu & 2 \delta_0{}^\mu & - 2
e_{\bmzero\bmzero}{}^\mu & 0 \\ 0 & 0 & 2 e_{\bmone\bmone}{}^\mu & 2
\delta_0{}^\mu & -2 e_{\bmzero\bmzero}{}^\mu \\ 0 & 0 & 0 & 2
e_{\bmzero\bmzero}{}^\mu & \delta_0{}^\mu -2 e_{\bmzero\bmone}{}^\mu
\end{array} \right) \partial_\mu \left(
\begin{array}{c} \phi_0 \\ \phi_1 \\ \phi_2 \\ \phi_3 \\ \phi_4
\end{array} \right).
\label{PrincipalPart}
\end{equation}

It follows from the above expression that the matrices associated to
the Bianchi subsystem are Hermitian and, moreover, $\mathbf{I}+
\mathbf{A}^0$ is positive definite. This is what one would expect from
a symmetric hyperbolic system.

It is worth stressing the remarkable structure of the system
\eqref{SchematicEvolutionSystem1}-\eqref{SchematicEvolutionSystem2}. The
subsystem \eqref{SchematicEvolutionSystem1} is a set of transport
equations along the conformal geodesics determining the gauge while
\eqref{SchematicEvolutionSystem2} is a symmetric hyperbolic system. In
what follows it will always be assumed that the Bianchi subsystem
\eqref{SchematicEvolutionSystem2} corresponds to the \emph{boundary
adapted} hyperbolic reduction discussed in Section
\ref{Subsubsection:SpaceSpinorSplit}. The relevance of this
requirement will become clear in Section
\ref{Section:PropOfConstraints}.

\subsection{Maximally dissipative boundary
conditions}\label{SubSection:MaximallyDissBCs} A key aspect in the
formulation of of an initial boundary value problem is the
identification of boundary conditions for which known existence
results of the theory of partial differential equations are
available. Our guide in this task is the theory of \emph{maximally
dissipative boundary conditions} ---see
e.g. \cite{kroon2017conformal}, Section 12.4.

\medskip The basic input in the framework of maximally dissipative
boundary conditions is the normal matrix $\mathbf{A}^3$ in the Bianchi
subsystem \eqref{SchematicEvolutionSystem2}. The latter can be read
from expression \eqref{PrincipalPart}. Taking into account the
properties of the boundary adapted conformal Gau\ss ian gauge system
introduced in Section \ref{Section:BoundaryGaugeConditions}, it
follows that
\begin{equation} \mathbf{A}^3|_{\mathcal{T}} = 2e_{\bmthree}{}^3
\left(
\begin{array}{ccccc} -1 & 0 & 0 & 0 & 0\\ 0 & 0 &0 &0 & 0 \\ 0 & 0 & 0
& 0 & 0 \\ 0 & 0 & 0 & 0 &0 \\ 0 & 0 & 0 &0 & 1
\end{array} \right).  \nonumber
\end{equation} As long as $e_\bmthree{}^3>0$ the dimension of the
Kernel of the above matrix remains constant so that the theory of
maximally dissipative boundary conditions can be applied. The latter
requires the condition
\begin{equation} \langle \bmphi,
\mathbf{A}^3|_{\mathcal{T}}\bmphi\rangle \leq 0, \nonumber
\end{equation} which is equivalent to
\begin{equation} |\phi_0|^2 -|\phi_4|^2 \leq 0.
\label{BCInequality}
\end{equation} Maximally dissipative boundary conditions are then
identified as the subspaces of $\mathbb{C}^5$ preserving condition
\eqref{BCInequality}. In order to characterise the subspaces of
$\mathbb{C}^5$ satisfying the above let $c_1$ and $c_2$ denote two
smooth complex-valued functions on $\mathcal{T}$ and let
\begin{equation} \phi_0 = c_1 \phi_4 + c_2 \bar{\phi}_4.
\end{equation} It follows then that
\begin{equation} |\phi_0|^2 -|\phi_4|^2 \leq (|c_1|^2 +|c_2|^2
-1)|\phi_4|^2.
\end{equation} Thus, condition \eqref{BCInequality} is satisfied if
one requires
\begin{equation} |c_1|^2 + |c_2|^2 \leq 1.  \nonumber
\end{equation} From the above it follows that suitable
\emph{inhomogeneous maximally dissipative boundary conditions} for our
system are given by the condition
\begin{equation} \phi_0 - c_1 \phi_4 - c_2 \bar{\phi}_4 = q, \qquad
|c_1|^2 +|c_2|^2 \leq 1, \nonumber
\end{equation} with $c_1$, $c_2$, $q$ smooth complex-valued functions
on the boundary $\mathcal{T}$. In the following, for simplicity we set
$c_1$ and restrict to the particular case
\begin{equation} \phi_0 -c\overline{\phi}_4= q , \qquad |c| \leq 1.
    \label{MaximallyDissipativeBCs}
\end{equation}

In general, the Bianchi evolution subsystem
\eqref{SchematicEvolutionSystem2} is a symmetric hyperbolic evolution
system containing four characteristic modes propagating normal to
$\T$: two lightlike propagating modes, one ingoing one outgoing and
two timelike propagating modes, one ingoing and one outgoing. The
theory of maximally dissipative boundary conditions allows the ingoing
characteristic modes to be prescribed a via boundary condition. The
ingoing lightlike mode can be thought of as the physical degree of
freedom in the Einstein equations, and as such, can be prescribed
freely.  However, as it will be seen in the following subsections, the
ingoing timelike mode is fixed by the requirement that the associated
mode in the constraint propagation system vanishes. This is a
complicated requirement, which is avoided by the use of the boundary
adapted system. This has the property that the evolution system no
longer has time-like propagating characteristic modes transverse to
$\T$, and the constraint propagation system has none. This has the
significant advantage that the only boundary conditions required are
those that can be prescribed freely.

\medskip The maximally dissipative boundary conditions
\eqref{MaximallyDissipativeBCs} with $c=1$ will be of particular
interest. In this case the boundary condition
\begin{equation} \phi_0-\bar{\phi}_4 =q
\label{InhomogeneousReflectiveBCs}
\end{equation} completely determines the value of the magnetic part
$B_{ABCD}$ of the Weyl tensor with respect to the normal to
$\mathcal{T}$. More precisely, from the discussion in
\ref{app:spinorformalism} it follows that the only non-trivial
component of the spinor $\mu_{ABCD}$ appearing in the irreducible
decomposition of $B_{ABCD}$ is $\mu_{\bmone\bmone\bmone\bmone}$. The
latter can be expressed, in terms of the components $\phi_4$ and
$\phi_0$ as
\begin{equation} \mu_{\bmone\bmone\bmone\bmone} =
-\frac{\mathrm{i}}{2}(\phi_4-\bar{\phi}_0), \nonumber
\end{equation} so that, in fact, one has
\begin{equation} \mu_{\bmone\bmone\bmone\bmone} = \frac{\mathrm{i}}{2}
\bar{q}.  \nonumber
\end{equation} As discussed in Section
\ref{Section:DecompositionBianchiEqns}, the spinor $\mu_{ABCD}$ acts
as a source of the transport equations
\eqref{BianchiTransport1}-\eqref{BianchiTransport2} on the timelike
boundary $\mathcal{T}$. Thus, if the geometry on $\mathcal{T}$ is
known, one can then solve these equations to obtain $\mu$ and
$\mu_{AB}$ and, thus, the rest of the spinor $B_{ABCD}$ encoding the
magnetic part of the the Weyl tensor with respect to the normal to
$\mathcal{T}$. Further consequences of the boundary condition
\eqref{InhomogeneousReflectiveBCs} will be elaborated in Section
\ref{Section:ConsistencyGauge}.


\subsubsection{Initial data.}  In addition to the boundary data
prescribed on the timelike boundary $\mathcal{T}$ one also needs to
prescribe initial data on the spacelike hypersurface
$\mathcal{S}_\star$. This data needs to be constructed in such a way
that it satisfies the conformal constraint ---namely the spacelike
analogue of the constraints on $\mathcal{T}$ ---equations
\eqref{eq:CCE-a}-\eqref{eq:CCE-j}. There are many ways to construct
this data ---in particular, given a solution to the usual Hamiltonian
and momentum constraints and a conformal factor describing the locus
of the boundary $\partial \mathcal{S}_\star$, there exists an
algebraic procedure to obtain the initial values of the fields
appearing in the conformal evolution equations
\eqref{SchematicEvolutionSystem1}-\eqref{SchematicEvolutionSystem2}. In
particular, this procedure ensures that the spacelike conformal
constraint equations are satisfied at the initial hypersurface ---this
is an important observation for the discussion of the propagation of
the constraints.  Given that the evolution system is first order,
information about the normal derivatives of the unknowns at
$\mathcal{S}_\star$ is required.

\subsubsection{Corner conditions.}
\label{Subsection:CornerConditions} The smoothness of a solution to an
initial boundary value problem requires certain compatibility
conditions (\emph{corner conditions}) between the initial data and the
boundary conditions at the intersection $\partial\mathcal{S}_\star
\equiv \mathcal{S}_\star\cap\mathcal{T}$. In principle, one can use
the boundary adapted Bianchi subsystem
\eqref{SchematicEvolutionSystem2} to determine formal expansions on
$\mathcal{T}$ near $\partial \mathcal{S}_\star$. These expansions, in
turn, imply an expansion for $\phi_0-c\bar{\phi}_4$ which must be
fixed with the prescription of the freely specifiable function $q$ in
\eqref{MaximallyDissipativeBCs}. The explicit form of these corner
conditions is rather cumbersome and no systematic treatment is
available in the literature ---it is an important open question. In
the following, and for the sake of presentation, it will be assumed
that these conditions are satisfied to any order.

\subsection{The local existence result} We have now all the
ingredients to provide a local existence statement for the initial
boundary value problem for the conformal evolution equations
\eqref{SchematicEvolutionSystem1}-\eqref{SchematicEvolutionSystem2}. This
result follows from the theory of maximally dissipative boundary
conditions for quasilinear symmetric hyperbolic systems ---see
e.g. Section 12.4 in \cite{kroon2017conformal} and references within.

\begin{proposition}\label{Proposition:LocalExistence}

Given an initial value problem for the conformal evolution equations
\eqref{SchematicEvolutionSystem1}-\eqref{SchematicEvolutionSystem2}
with smooth initial data $(\bmupsilon_\star,\bmphi_\star)$ on
$\mathcal{S}_\star$ and inhomogeneous maximally dissipative boundary
data
\begin{equation} \phi_0 -c \bar{\phi}_4 = q, \qquad |c|^2 \leq 1 \quad
\mbox{on} \quad \mathcal{T}, \nonumber
\end{equation} with $c$ and $q$ smooth complex-valued functions and
assuming that the corner conditions at $\partial\mathcal{S}_\star$
between initial and boundary conditions are satisfied to any order,
then there exists a $\tau_\bullet>0$ such that the initial boundary
value problem has a unique smooth solution defined on
\begin{equation} \mathcal{M}_{\tau_\bullet} \equiv [0,\tau_\bullet)
\times \mathcal{S}.  \nonumber
\end{equation}
\end{proposition}

It is important to stress that the above result does not assert the
existence of solutions to the full extended conformal Einstein field
---in particular, it does not establish the consistency of the gauge
and the assumption that the timelike boundary is rules by conformal
geodesics. These issues will be analysed in the remaining sections.

\section{Propagation of constraints}\label{Section:PropOfConstraints}
Given a solution to the conformal evolution equations
\eqref{SchematicEvolutionSystem1}-\eqref{SchematicEvolutionSystem2},
as given by Proposition \ref{Proposition:LocalExistence}, we now
analyse the conditions ensuring that this solution implies, in turn, a
solution to the full extended conformal Einstein field equations. This
analysis is often called the \emph{propagation of the
constraints}. The analysis of the propagation of the constraints is
the central concern in this article.

\medskip As is usual in this type of analysis, we consider the
\emph{zero-quantities} associated to the various equations in the
extended conformal Einstein field equations and the gauge conditions
(the conformal Gau\ss ian gauge) used in the derivation of the
evolution equations. In the following, let
\begin{eqnarray*} && \bmUpsilon: \qquad \mbox{be the zero-quantities
associated to the conformal equations for $\bmupsilon$, }\\ &&
\bmSigma: \qquad \mbox{be the zero-quantities associated to the
conformal Gau\ss ian gauge},\\ && \bmLambda: \qquad \mbox{be the
zero-quantities associated to the Bianchi equations.}
\end{eqnarray*} In \cite{kroon2017conformal}, it has been shown
(following \cite{friedrich1995einstein}) that, in the F-gauge, the
zero-quantities satisfy an evolution system (the \emph{subsidiary
system}) with the following properties:
\begin{subequations}
\begin{eqnarray} && \partial_\tau {\bm \Upsilon} =
\mathcal{H}_1(\bmUpsilon,\bmSigma,\bmLambda),\label{BoundaryAdaptedSubsydiarySystem1}\\
&& \partial_\tau \bmSigma
=\mathcal{H}_2(\bmSigma,\bmUpsilon),\label{BoundaryAdaptedSubsydiarySystem2}\\
&& \mathbf{B}^0(\bmUpsilon) \partial_\tau \bmLambda +
\mathbf{B}^{\mathcal{A}}(\bmUpsilon)\partial_{\mathcal{A}}\bmLambda =
\mathcal{H}_3(\bmUpsilon,\bmSigma,\bmLambda), \qquad \mathcal{A} =
1,2\label{BoundaryAdaptedSubsydiarySystem3}
\end{eqnarray}
\end{subequations} where $\mathcal{H}_1$, $\mathcal{H}_2$ and
$\mathcal{H}_3$ are homogeneous functions of their arguments in the
sense that
\begin{equation} \mathcal{H}_1(0,0,0)=0, \qquad \mathcal{H}_2(0,0)=0,
\qquad \mathcal{H}_3(0,0,0)=0, \nonumber
\end{equation} and the matrices $\mathbf{B}^0$,
$\mathbf{B}^{\mathcal{A}}$ are Hermitian ---in particular,
$\mathbf{B}^0$ is definite positive so that
\eqref{BoundaryAdaptedSubsydiarySystem3} constitutes a symmetric
hyperbolic subsystem.  The structure of the subsidiary system
\eqref{BoundaryAdaptedSubsydiarySystem1}-\eqref{BoundaryAdaptedSubsydiarySystem3}
is reminiscent of that of the conformal evolution system
\eqref{SchematicEvolutionSystem1}-\eqref{SchematicEvolutionSystem2} in
that it consists of transport equations along the conformal geodesics
coupled to a symmetric hyperbolic system.

The central structural property of not containing derivatives in the
direction transversal to the boundary can be heuristically understood
as follows: when considering propagation transversal to the boundary,
the three non-null characteristic modes of the evolution system match
the three characteristic modes of the corresponding subsidiary
system. Hence, as the boundary adapted evolution system (in our case,
the boundary adapted evolution system for the spin-2 field together
with Weyl propagation of the frame and a restriction on the geometry
of $\T$) has all three non-null modes zero, the subsidiary system then
no longer propagates in directions transversal to the boundary. This
property has the consequence that for the boundary adapted evolution
system, the only ingoing characteristic mode that needs to be dealt
with is light-like, which corresponds to the single complex physical
degree of freedom inherent in the Einstein equations. No tie between
ingoing modes of the evolution system and ingoing modes of the
subsidiary system need to be made to ensure that the imposition of
boundary data satisfy the constraints.

\medskip Crucially for our analysis, the hyperbolic system does not
contains derivatives in the direction transversal to the timelike
boundary $\mathcal{T}$. Accordingly, the normal matrix vanishes so
that it is not possible to prescribe boundary conditions for the
system \eqref{BoundaryAdaptedSubsydiarySystem3} ---in other words, the
solutions to the subsidiary system
\eqref{BoundaryAdaptedSubsydiarySystem1}-\eqref{BoundaryAdaptedSubsydiarySystem3}
can be fully determined from the initial data at the corner $\partial
\mathcal{S}_\star$. This is a crucial property which does not hold for
other choices of evolution equations for the Bianchi equations.

\medskip From the previous discussion it follows, in particular, that
if one has
\begin{equation} \bmUpsilon|_{\mathcal{S}_\star}=0, \qquad
\bmSigma|_{\mathcal{S}_\star}=0, \qquad
\bmLambda|_{\mathcal{S}_\star}=0, \nonumber
\end{equation} then, by uniqueness, one has that
\begin{equation} \bmUpsilon=0, \qquad \bmSigma=0, \qquad \bmLambda=0
\quad \mbox{at later times.}  \nonumber
\end{equation}

\medskip We can summarise the above discussion in the following:

\begin{proposition} [Propagation of the constraints] A solution to the
conformal evolution equations
\eqref{SchematicEvolutionSystem1}-\eqref{SchematicEvolutionSystem2}
with initial data on an initial hypersurface $\mathcal{S}_\star$
satisfying the conformal constraint equations and maximally
dissipative boundary conditions on the timelike hypersurface
$\mathcal{T}$ determines a solution to the conformal Einstein field
equations in the domain of existence of the the solution to the
conformal evolution equations.
\end{proposition}

It is worth observing that the properties of the subsidiary evolution
system are similar to those of the intrinsic equations implied by the
conformal Einstein field equations on the cylinder at spatial infinity
---see \cite{Fri98a}. In that situation one has, again, a setting with
a timelike boundary which does not admit the prescription of boundary
data.

\section{Consistency of the gauge and the geometric setting }
\label{Section:ConsistencyGauge} After having identified boundary
conditions, discussed the local solvability of the conformal evolution
system and the propagation of the constraints, in this last section,
we discuss the consistency between the conformal Gau\ss ian gauge
adapted to the timelike boundary $\mathcal{T}$ and the maximally
dissipative boundary conditions.

\medskip The key ingredient in the construction of our adapted gauge
is summarised in Assumption
\ref{Assumption:BoundaryRuledByConformalGeodesics} and requires that
$\mathcal{T}$ is ruled by a nonintersecting congruence of conformal
geodesics. In turn, Lemma \ref{Lemma:IntrinsicCG} shows that this
assumption imposes restrictions on the (intrinsic) curvature of
$\mathcal{T}$. As discussed in Section
\ref{Section:ConformalConstraints} the conformal constraints on
$\mathcal{T}$ relate the intrinsic curvature of $\mathcal{T}$ to the
curvature of the spacetime ---in particular, the rescaled Weyl tensor
$d^a{}_{bcd}$. It follows form the discussion in Section
\ref{Section:Decomposition:GaussConstraints} that, to some extent, the
maximally dissipative conditions determine the value of the components
of the Weyl tensor on $\mathcal{T}$. Thus, it is \emph{a priori} not
clear the extent to which the maximally dissipative boundary
conditions for the conformal evolution system are consistent with the
geometric requirements of Assumption
\ref{Assumption:BoundaryRuledByConformalGeodesics}. In this section we
identify a subclass of initial boundary conditions, namely that given
by equation \eqref{InhomogeneousReflectiveBCs} for which it is
possible to show the consistency with the adapted conformal Gaussian
gauge.

\subsection{Reconstructing $B_{ABCD}$ on $\mathcal{T}$}

Following the discussion in Section \ref{SubSection:MaximallyDissBCs},
in the following we consider the boundary condition
\eqref{MaximallyDissipativeBCs} with $c=1$, so that one has
\begin{equation} \phi_0-\bar{\phi}_4 =q \nonumber
\end{equation} ---cfr. equation \eqref{InhomogeneousReflectiveBCs}. In
the particular case where $q=0$ one has that $\phi_0=\bar{\phi}_4$
---that is, one has fully reflective boundary conditions.

\medskip Now, as already seen, the boundary condition
\eqref{InhomogeneousReflectiveBCs} completely determines
$\mu_{\bmone\bmone\bmone\bmone} $ ---the only nontrivial component of
the spinor $\mu_{ABCD}$ which is part of the irreducible demposition
of the magnetic part spinor $B_{ABCD}$. More precisely, one has that
\begin{equation} \mu_{\bmone\bmone\bmone\bmone} = \frac{\mathrm{i}}{2}
\bar{q}.  \nonumber
\end{equation} As discussed in Section
\ref{Section:DecompositionBianchiEqns}, the spinor $\mu_{ABCD}$ acts
as a source of the transport equations
\eqref{BianchiTransport1}-\eqref{BianchiTransport2} on the timelike
boundary $\mathcal{T}$. Thus, if the geometry on $\mathcal{T}$ is
known, one can then solve these equations to obtain $\mu$ and
$\mu_{AB}$ and, thus, the rest of the spinor $B_{ABCD}$ encoding the
magnetic part of the the Weyl tensor with respect to the normal to
$\mathcal{T}$. This is a consequence of the fact that the evolution
system intrinsic to $\mathcal{T}$ given by equations
\eqref{BianchiTransport1}-\eqref{BianchiTransport2} is symmetric
hyperbolic so one has local existence of solutions for data prescribed
on $\partial\mathcal{S}_\star$. The initial data for $\mu$ and
$\mu_{AB}$ can be readily read off from the restriction to the corner
of the data on $\mathcal{S}_\star$ for the conformal evolution
equations
\eqref{SchematicEvolutionSystem1}-\eqref{SchematicEvolutionSystem2}. The
initial data on $\mathcal{S}_\star$ also determines the value of
$\mu_{ABCD}$ on $\partial\mathcal{S}_\star$. The corner conditions
discussed in Subsection \ref{Subsection:CornerConditions} ensure that
this value is consistent with the boundary condition
\eqref{InhomogeneousReflectiveBCs}.

\medskip We summarise the previous discussion in the following:

\begin{proposition}[Reconstructing the magnetic part of the Weyl
tensor from boundary conditions] Assuming that the required corner
conditions between initial data at $\mathcal{S}_\star$ and boundary
conditions at $\mathcal{T}$ are satisfied, the spinor $B_{ABCD}$ is
completely determined, in a neighbourhood of $\partial
\mathcal{S}_\star$ on $\mathcal{T}$, by the boundary conditions
\eqref{InhomogeneousReflectiveBCs} and the initial value of the
rescaled Weyl tensor at $\partial\mathcal{S}_\star$.
\end{proposition}

\subsection{Reconstructing the geometry of $\mathcal{T}$} Having
reconstructed the electric part of the Weyl tensor on $\mathcal{T}$ we
now show how the boundary conditions also fix the intrinsic geometry
of the hypersurface $\mathcal{T}$.

\medskip In order to reconstruct the geometry, we proceed as in
Section 17.3.2 of \cite{kroon2017conformal}, page 472.  Accordingly,
in the following let ${\bm h}$ denote the metric of the timelike
boundary $\mathcal{T}$. Moreover, let $\hat{D}$ denote a Weyl
connection in the conformal class of ${\bm h}$. Schematically, the
Levi-Civita connection of ${\bm h}$ and the Weyl connection $\hat{D}$
are related to each other by
\begin{equation} \hat{D} -D = {\bm S}(\bmf) \nonumber
\end{equation} where $\bm f$ is a 3-dimensional covector over
$\mathcal{T}$. In addition to the above, let $\{ \bme_\bmi \}$ be a
${\bm h}$-orthogonal frame on $\mathcal{T}$ and denote by
$\hat{\gamma}_\bmi{}^\bmj{}_\bmk$ the associated connection
coefficients of $\hat{D}$. One then has the \emph{structure equations}
\begin{eqnarray*} \hat{\Sigma}_\bmi{}^\bmk{}_\bmj \bme_\bmk =0, \qquad
\hat{\Xi}^\bmk{}_{\bml\bmi\bmj}=0, \qquad
\hat{\Delta}_{\bmi\bmj\bmk}=0, \qquad \Lambda_\bmj =0,
\end{eqnarray*} in terms of the zero-quantities
\begin{eqnarray*} && \hat{\Sigma}_\bmi{}^\bmk{}_\bmj \bme_\bmk \equiv
[\bme_\bmi,\bme_\bmj] -( \hat{\gamma}_\bmi{}^\bmk{}_\bmj
-\hat{\gamma}_\bmj{}^\bmk{}_\bmi)\bme_\bmk,\\ &&
\hat{\Xi}^\bmk{}_{\bml\bmi\bmj} \equiv
\bme_\bmi(\hat{\gamma}_\bmj{}^\bmk{}_\bml)-\bme_\bmj(\hat{\gamma}_\bmi{}^\bmk{}_\bml)+
\hat{\gamma}_\bmm{}^\bmk{}_\bml (\hat{\gamma}_\bmj{}^\bmm{}_\bmi -
\hat{\gamma_\bmi{}^\bmm{}_\bmj})\\ &&
\hspace{3cm}+\hat{\gamma}_\bmj{}^\bmm{}_\bml
\hat{\gamma}_\bmi{}^\bmk{}_\bmm
-\hat{\gamma}_\bmi{}^\bmm{}_\bml\hat{\gamma}_\bmj{}^\bmk{}_\bmm -2
S_{\bml[\bmi]}{}^{\bmk\bmm}\hat{l}_{\bmj]\bmm},\\ &&
\hat{\Delta}_{\bmi\bmj\bmk} \equiv \hat{D}_\bmi \hat{l}_{\bmj\bmk} -
\hat{D}_\bmj l_{\bmi\bmk} -y_{\bmi\bmj\bmk}, \\ && \Lambda_{\bmj}
\equiv D^\bmi y_{\bmi\bmj},
\end{eqnarray*} where it is recalled that the \emph{Bach tensor} is
given by
\begin{equation} y_{ij} \equiv -\frac{1}{2}\epsilon_j{}^{kl}y_{ikl},
\qquad y_i{}^i=0, \quad y_{ij}=y_{ji}, \nonumber
\end{equation} where $y_{ikl}$ is the Cotton tensor ---see
equation~\eqref{eq:cotton}. By construction we are requiring that the
timelike boundary $\mathcal{T}$ is ruled by conformal
geodesics. Accordingly one has the associated gauge conditions (see
Section \ref{Section:BoundaryGaugeConditions})
\begin{equation} e_\bmzero{}^\alpha =\delta_\bmzero{}^\alpha, \qquad
\hat{\gamma}_\bmzero{}^\bmk{}_\bmj =0, \qquad \hat{l}_{\bmzero
\bmj}=0, \nonumber
\end{equation} together with the evolution equations
\begin{equation} \hat{\Sigma}_\bmzero{}^\bmk{}_\bmj \bme_\bmk =0,
\qquad \hat{\Xi}^\bmk{}_{\bml\bmzero\bmj}=0, \qquad
\hat{\Delta}_{\bmzero\bmj\bmk}=0, \qquad \Lambda_\bmj=0.
\label{EvolutionEqns:StructureBoundary}
\end{equation} It is observed that in this gauge one also has \emph{a
priori} knowledge of the canonical conformal factor associated to the
congruence of conformal geodesics ---to be denoted by $\Omega$. The
argument behind this last statement is analogous to the one used to
fix the spacetime conformal factor through the spacetime conformal
Gaussian gauge ---see e.g. Proposition 5.1 in
\cite{kroon2017conformal}.

\medskip Now, an important observation is that $y_{ij}$ (alternatively
$y_{ijk}$) is related to $d_{ijk}$. More precisely, from the conformal
constraint equation \eqref{eq:CCE-d} recalling that
\begin{equation} L_{ij} =l_{ij}, \qquad L_i=0 \nonumber
\end{equation} one concludes that
\begin{equation} y_{ijk} =\Sigma d_{kij} +D^l\Omega d_{lkij}, \qquad
d_{ijkl} = 2 \big(l_{i[k}] d_{l]j} +l_{j[l} d_{k]i} \big).  \nonumber
\end{equation} Now, the condition $L_{ij}=l_{ij}$ fixes $d_{ij}$ in
terms of $K_{ij}$ ---namely, one has that
\begin{equation} \Omega d_{ij} = K \Big( K_{ij}-\frac{1}{4}
Kh_{ij}\Big) - K_{ki}K_j{}^k +\frac{1}{4}K_{kl}K^{kl}l_{ij}.
\label{ComputingElectricPart}
\end{equation} Accordingly, one still needs to obtain a suitable
equation for the components of $K_{ij}$. From the Gauss-Codazzi
equation one concludes that (by taking traces) that $K=0$ so that
$K_{ij}=K_{\{ij\}}$ ---that is, the extrinsic curvature of the
timelike hypersurface $\mathcal{T}$ is trace-free. Thus, we have found
the following restriction on the extrinsic geometry of $\mathcal{T}$.

\begin{lemma} The hypersurface $\mathcal{T}$ is maximal.
\end{lemma}

The strategy is now to construct a suitable evolution equation on
$\mathcal{T}$ for the remaining components of $K_{ij}$. To this end we
consider the equation
\begin{equation} D_j K_{ki} - D_k K_{ji} =\Omega d_{ijk}.
\label{CodazziMainardiALT}
\end{equation} In this equation the Cotton-York tensor acts a
source. A convenient way of carrying out the required hyperbolic
reduction is to make use of space spinors.

\medskip Given that $K_{ij}$ is trace-free, its space spinor
counterpart is described by a totally symmetric spinor
$K_{ABCD}=K_{(ABCD)}$, and one obtains the equation
\begin{equation} D_{AB}K_{CDEF} - D_{CD}K_{ABEF} = \Omega d_{ABCDEF}.
\nonumber
\end{equation} From the later it follows, after contraction of the
indices $A$ and $C$ that
\begin{equation} D^Q{}_{(B}K_{D)QEF} = \frac{1}{2}\Omega d^Q{}_{BQEF}.
\nonumber
\end{equation} Following the general theory for hyperbolic reductions
in spinor formalism discussed in \cite{kroon2017conformal} one needs
to break the symmetrisation in this equation by adding the trace
$D^{PQ}K_{PQEF}$. Notice, however, that from equation
\eqref{CodazziMainardiALT} taking the trace over $ij$ it readily
follows that $K=0$. Accordingly one has that $D^Q{}_{(B}K_{D)QEF} =
D^Q{}_B K_{DQEF}$ so that, in the end, one has that
\begin{equation} D^Q{}_{(B}K_{D)QEF} = \frac{1}{2}\Omega d^Q{}_{BQEF}.
\nonumber
\end{equation}

\medskip Now, making use of the decomposition $D_{AB} =
\frac{1}{2}\tau_{AB}\partial + \partial_{AB} $ and using that
$\tau^B{}_A \tau^Q{}_B=\epsilon_A{}^Q$ one can conclude that
\begin{equation} \partial K_{ADEF} +\tau^B{}_A \partial^Q{}_B K_{DQEF}
= \frac{1}{2}\Omega\tau^B{}_A d^Q{}_{BQEF}.
\label{PropagationKBoundary}
\end{equation} The above equation constitute a symmetric hyperbolic
system for the independent components of $K_{ABCD}$ which can be
locally solved if the source $\Omega d^Q{}_{BQEF}$ is known.

\subsubsection*{Concluding the argument.} The lengthy discussion in
the previous subsection has lead to a coupled system consisting of the
subsystems \eqref{EvolutionEqns:StructureBoundary},
\eqref{BianchiTransport1}-\eqref{BianchiTransport2}, and
\eqref{PropagationKBoundary} for the fields
\begin{equation} B_{ABCD} (\Leftrightarrow d^*_{ij} \Leftrightarrow
d_{ijk}) , \quad l_{ij} \;\;\mbox{(together with its connection and
curvature)}, \quad K_{ij}.  \nonumber
\end{equation}

The intrinsic evolution equations on $\mathcal{T}$ satisfied by these
fields constitute a symmetric hyperbolic system for which, again,
standard theory implies the existence of a unique solution in a
neighbourhood of $\partial\mathcal{S}_\star$ in $\mathcal{T}$. The
fields obtained from this solution allow to compute the components of
the electric part $d_{ij}$ using equation
\eqref{ComputingElectricPart}. Observe that the conformal Gau\ss ian
gauge used in $\mathcal{T}$ gives direct \emph{a priori} knowledge of
$\Omega$.

\medskip At this point it only remains to carry out a
\emph{propagation of the constraints argument} for the unused
structure equations on $\mathcal{T}$. This can be shown by means of an
argument analogous to the one used to show the propagation of the
constraints associated to the conformal Einstein field equations. We
omit the details.

\medskip In summary, the previous discussion shows that the particular
class of maximally dissipative boundary conditions given by
\eqref{InhomogeneousReflectiveBCs} together with the initial data at
the corner $\partial\mathcal{S}_\star$ determine the intrinsic and
extrinsic geometry of the timelike boundary $\mathcal{T}$. Further,
the discussion also shows that \eqref{InhomogeneousReflectiveBCs} is
consistent with the geometric restrictions imposed by Assumption
\ref{Assumption:BoundaryRuledByConformalGeodesics}.

\section{Conclusions and outlook} In this work we have established a
wellposed initial boundary value problem for the extended conformal
Einstein field equations in the presence of a timelike boundary
located at finite distance. Building on a conformal Gau\ss ian gauge
adapted to a congruence of timelike conformal geodesics, and
introducing a boundary-adapted spinorial formalism, we have shown that
the system admits a symmetric hyperbolic reduction compatible with
maximally dissipative boundary conditions. Crucially, by employing a
boundary-adapted reduction of the Bianchi subsystem, the
characteristic structure can be made to simplify so that only a single
physical mode propagates transversely to the boundary. This enables a
clean prescription of boundary data in terms of a single complex free
datum, directly corresponding to the physical degrees of freedom of
the gravitational field.

A key outcome of the analysis is that the subsidiary system governing
the propagation of the constraints reduces to a transport system
intrinsic to the boundary. As a consequence, constraint preservation
follows without the need for additional boundary
conditions. Furthermore, by analysing the conformal constraint
equations intrinsic to the boundary hypersurface, we have demonstrated
that the remaining components of the Weyl tensor are fully determined
by the maximally dissipative boundary conditions together with the
intrinsic geometry of the boundary. Taken together, these results
provide a rigorous mathematical underpinning for a class of initial
boundary value problems of the conformal Einstein field equations.

\medskip The present work opens several directions for further
investigation.

\medskip First, the analysis has relied on the use of a
boundary-adapted evolution system for the Bianchi equations, which
eliminates spurious incoming modes and leads to a particularly simple
characteristic structure. Previous numerical implementations of the
conformal field equations have not employed this boundary-adapted
reduction \cite{beyer2017numerical,frauendiener2021non,
frauendiener2023non, frauendiener2023non2, frauendiener2025fully}, yet
nevertheless exhibit stable and convergent behaviour. This strongly
suggests that an analogous wellposedness result may hold in the more
general setting. Establishing such a result would require a careful
analysis of the additional ingoing modes and their interplay with the
subsidiary system.

Second, the class of boundary conditions considered here has been
restricted to maximally dissipative conditions, in particular the
fully reflective case. While this choice is natural from the
perspective of wellposedness, it is of clear physical and numerical
interest to investigate whether these conditions can be relaxed. In
particular, allowing for partially transmitting or radiative boundary
conditions could provide a more flexible framework for modelling
physically realistic scenarios, especially in the context of
approximating gravitational backreaction.

Finally, our construction has relied on a geometric restriction on the
boundary hypersurface to enforce propagation of the boundary-adapted
frame conditions. An alternative approach is to abandon this geometric
restriction and instead modify the frame propagation equation. This
strategy introduces additional gauge freedom which can be used to
enforce the desired characteristic structure without constraining the
geometry of the boundary. This approach has already been explored in
related work and has proved to perform well both analytically and
numerically. These results will be published soon
\cite{frauendiener2025fullynonlinear}.

Overall, the results presented here contribute to bridging the gap
between the geometric formulation of the conformal Einstein field
equations and their practical implementation as initial boundary value
problems. They provide a robust foundation for future analytical and
numerical developments aimed at understanding the global properties of
spacetimes through conformal methods.

\section*{Acknowledgements} This work was supposed by the Marsden Fund
Council from government funding managed by the Royal Society Te
Apārangi of New Zealand.  JAVK thanks the hospitality of the Erwin
Schr\"odinger Institute for Mathematics and Physics of the University
of Vienna where part of this research was carried out as part of the
workshop \emph{Carrollian Physics and Holography} in April 2024.

\appendix
\section{A 1+1+2 spinor formalism}\label{app:spinorformalism}

A natural step when studying an Initial Boundary Value Problem (IBVP)
formulation of a system of equations is to understand the structure of
the equations projected onto a spacelike initial hypersurface and a
timelike boundary hypersurface. For the below discussion, we consider
an IBVP formulation of the Conformal Field Equations (CFE), where
these projections are well known ---see equation (11.35) of
\cite{kroon2017conformal}. These equations yield a 4-dimensional
space-time $(\mathcal{M},\bmg)$ with Lorentzian metric $\bmg$. On this
space-time, we have the notion of \emph{2-spinors}. Here we only give
basic properties and refer the reader to \cite{penrose1984spinors} or
\cite{kroon2017conformal} for an in-depth description.

\subsection*{$SL(2,\mathbb{C})$ spinors}

The $SL(2,\mathbb{C})$-spinor formalism is built from elements of a
2-dimensional complex vector space $\mathfrak{S}^A$, its conjugate
vector space $\mathfrak{S}^{A'}$ and their dual spaces
$\mathfrak{S}_A$ and $\mathfrak{S}_{A'}$. Given an orthonormal frame
$\{\bm{e}_{\bm{a}}\}$ for $\bmg$ so that ${\eta}_{\bma\bmb} =
\bm{g}(\bm{e}_{\bm{a}},\bm{e}_{\bm{b}})$ where $({\eta}_{\bma\bmb})
\equiv \text{diag}(1,-1,-1,-1)$, and $\{\bm{\epsilon}_A\}$ a
spin-frame, then one has $\epsilon_{\bm{AB}}\epsilon_{\bm{A'B'}} =
\sigma^{\bm{a}}{}_{\bm{AA'}}\sigma^{\bm{b}}{}_{\bm{BB'}}\eta_{\bm{ab}}$
and from this one can then obtain relations between general tensors
and their spinorial counterparts. For example, one can relate frame
components of a vector to spin-frame components of a spinor via
\begin{equation} (v^{\bm{a}}) =
(v^{\bm{0}}\;v^{\bm{1}}\;v^{\bm{2}}\;v^{\bm{3}}) \mapsto
(v^{\bm{AA'}}) = \frac{1}{\sqrt{2}}
    \begin{pmatrix} v^{\bm{0}} + v^{\bm{3}} & v^{\bm{1}} + i
v^{\bm{2}} \\ v^{\bm{1}} - i v^{\bm{2}} & v^{\bm{0}} - v^{\bm{3}}
    \end{pmatrix}.\nonumber
\end{equation} It is the \emph{Infeld-van der Waerden symbols}
$\sigma_{\bm{a}}^{\bm{AA'}}{}$ and $\sigma^{\bm{a}}_{\bm{AA'}}{}$ that
translate tensor components to spinor components. These are simply (up
to normalisation) the Pauli matrices
\begin{gather*} (\sigma_{\bm{0}}{}^{\bm{AA'}}) =
(\sigma^{\bm{0}}{}_{\bm{AA'}}) \equiv \frac{1}{\sqrt{2}}
    \begin{pmatrix} 1 & 0 \\ 0 & 1
    \end{pmatrix},\quad (\sigma_{\bm{1}}{}^{\bm{AA'}}) =
(\sigma^{\bm{1}}{}_{\bm{AA'}}) \equiv \frac{1}{\sqrt{2}}
    \begin{pmatrix} 0 & 1 \\ 1 & 0
    \end{pmatrix},\\[8pt] (\sigma_{\bm{2}}{}^{\bm{AA'}}) =
-(\sigma^{\bm{2}}{}_{\bm{AA'}}) \equiv \frac{1}{\sqrt{2}}
    \begin{pmatrix} 0 & i \\ -i & 0
    \end{pmatrix},\quad (\sigma_{\bm{3}}{}^{\bm{AA'}}) =
(\sigma^{\bm{3}}{}_{\bm{AA'}}) \equiv \frac{1}{\sqrt{2}}
    \begin{pmatrix} 1 & 0 \\ 0 & -1
    \end{pmatrix}.
\end{gather*}

It is clear that the $SL(2,\mathbb{C})$-spinor of a 4-vector with real
components is Hermitian.

\subsection*{$SU(2)$ spinors} Assuming the existence of a timelike
vector field $\tau^a$ normalised as $\tau^a\tau_a=2$, a map can then
be constructed $\mathfrak{S}_{A'}\ni\alpha_{A'} \mapsto
\alpha_A\in\mathfrak{S}_A$ via
\begin{equation} \alpha_{A'} \mapsto \alpha_A \equiv
\tau_{A}{}^{A'}\alpha_{A'}, \nonumber
\end{equation} and similarly for contravariant indices. With this, one
need only consider unprimed spinor indices and the result is the
\emph{space-spinor formalism} \cite{sommers1980space}. The relation
$\tau_{AA'}\tau_B{}^{A'} = \epsilon_{AB}$ identifies the antisymmetric
parts of an unprimed spinor with information along $\tau^a$ and the
symmetric parts with information orthogonal to $\tau^a$. As the
antisymmetric parts can always be written as a trace multiplied by the
$\epsilon$ spinor, the symmetric (and hence spatial) pieces are the
only irreducible pieces with indices remaining ---the latter motivates
the name of the formalism. The relevance here is, given a spinorial
differential equation from the CFE, the procedure of mapping primed to
unprimed indices with $\tau^a$ and performing an irreducible
decomposition is a mechanical way of obtaining evolution and
constraint equations \cite{magdy2025space}.

\medskip Now, defining a spin-frame $\{o^A \equiv
\epsilon_{\bm{0}}{}^A, \iota^A \equiv \epsilon_{\bm{1}}{}^A\}$ with
$o_A\iota^A=1$ adapted to $\tau^a$ through the condition
\begin{equation} \tau^{AA'} = o^Ao^{A'} + \iota^A\iota^{A'} \qquad
\mbox{so that}\quad (\tau^{\bm{AA'}}) =
    \begin{pmatrix} 1 & 0 \\ 0 & 1
    \end{pmatrix}, \nonumber
\end{equation} the $SU(2)$ spinor mappings of the Infeld-van der
Waerden symbols are given by
\begin{equation} \sigma^{i}{}_{AB} \equiv
\tau_A{}^{A'}\sigma^{i}{}_{BA'}\qquad \mbox{and} \qquad
\sigma_{i}{}^{AB} \equiv -\tau^A{}_{A'}\sigma_{i}{}^{BA'}, \nonumber
\end{equation} with components
\begin{gather*} (\sigma_{\bm{0}}{}^{\bm{AB}}) =
(\sigma^{\bm{0}}{}_{\bm{AB}}) \equiv \frac{1}{\sqrt{2}}
    \begin{pmatrix} 0 & 1 \\ -1 & 0
    \end{pmatrix},\quad (\sigma_{\bm{1}}{}^{\bm{AB}}) =
(\sigma^{\bm{1}}{}_{\bm{AB}}) \equiv \frac{1}{\sqrt{2}}
    \begin{pmatrix} -1 & 0 \\ 0 & 1
    \end{pmatrix},\\[8pt] (\sigma_{\bm{2}}{}^{\bm{AB}}) =
-(\sigma^{\bm{2}}{}_{\bm{AB}}) \equiv \frac{1}{\sqrt{2}}
    \begin{pmatrix} -i & 0 \\ 0 & -i
    \end{pmatrix},\quad (\sigma_{\bm{3}}{}^{\bm{AB}}) =
(\sigma^{\bm{3}}{}_{\bm{AB}}) \equiv \frac{1}{\sqrt{2}}
    \begin{pmatrix} 0 & 1 \\ 1 & 0
    \end{pmatrix}.
\end{gather*} As the $SU(2)$ spinor formalism only involves totally
symmetric spinors, the symbols $\sigma_{\bm{0}}{}^{\bm{AB}} =
\sigma^{\bm{0}}{}_{\bm{AB}}$ make no contribution to the resulting
$SU(2)$ spinor components and $\sigma_{\bm{i}}{}^{\bm{AB}} =
\sigma^{\bm{i}}{}_{\bm{AB}}$, $i=1,2,3$ are then the space-spinor
Infeld-van der Waerden symbols. Thus, the components of a vector can
be translated into components of a $SU(2)$ spinor through the
correspondence
\begin{equation} (v^{\bm{a}}) =
(v^{\bm{0}}\;v^{\bm{1}}\;v^{\bm{2}}\;v^{\bm{3}}) \mapsto (v^{\bm{AB}})
= \frac{1}{\sqrt{2}}
    \begin{pmatrix} -v^{\bm{1}} - iv^{\bm{2}} & v^{\bm{3}} \\
v^{\bm{3}} & v^{\bm{1}} - iv^{\bm{2}}
    \end{pmatrix}.  \nonumber
\end{equation} The space-spinor of a real 4-vector is no longer
Hermitian. In fact, defining complex conjugation by
\begin{equation} \hat{v}^{AB} = -
\bar{v}^{A'B'}\tau^A{}_{A'}\tau^B{}_{B'}, \nonumber
\end{equation} it is found that the reality of $v^{\bm{a}}$ requires
$\hat{v}^{\bm{AB}} = -v^{\bm{AB}}$.

\subsection*{$SU(1,1)$ spinors} The $SU(2)$ spinor formalism is
extremely useful for projecting spinorial expressions along and
orthogonal to some time-like vector field $\tau^a$. However, when
considering a time-like boundary $\T$, such as when considering an
IBVP, it is natural to do a similar splitting, but now with respect to
a space-like vector $\rho^a$ normal to $\T$. One then should end up
with unprimed spinors whose antisymmetric parts are along $\rho^a$ and
symmetric parts are intrinsic to $\T$.

\medskip Proceeding analogously to the case of $SU(2)$ spinors, we
choose our spin-frame so that
\begin{equation} \rho^{AA'} = o^Ao^{A'} - \iota^A\iota^{A'} \nonumber
\end{equation} with $\rho^a\rho_a=-2$. Then we define our map from
$\mathfrak{S}_{A'}$ to $\mathfrak{S}_A$ via
\begin{equation} \alpha_{A'} \mapsto \alpha_A \equiv
\mathrm{i}\rho_{A}{}^{A'}\alpha_{A'}, \nonumber
\end{equation} and similarly for contravariant indices. The important
difference in using a spatial vector as a map from primed to unprimed
spinors is that it requires a factor of $\mathrm{i}$ to preserve the
proper sign of the spacetime interval, preserving spacelike and
timelike norms. The relation $\rho_B{}^{A'}\rho_{AA'} =
-\mathrm{i}\epsilon_{AB}$ holds and we maintain the correct sign in
the norm of $\rho_{AB}$. Further, in direct analogy to $SU(2)$
spinors, the antisymmetric parts of a $SU(1,1)$ spinor expression
correspond to information along $\rho^a$ and symmetric parts
correspond to information orthogonal to $\rho^a$. The spin-frame
expression for $\rho^{AA'}$ is
\begin{equation} \rho^{AA'} = o^Ao^{A'} - \iota^A\iota^{A'} \qquad
\mbox{so that }\quad (\rho^{\bm{AA'}}) =
    \begin{pmatrix} -1 & 0 \\ 0 & 1
    \end{pmatrix} \nonumber
\end{equation} and the $SU(1,1)$ spinor mappings of the Infeld-van der
Waerden symbols are given by
\begin{equation} \sigma^{i}{}_{AB} \equiv
\mathrm{i}\rho_A{}^{A'}\sigma^{i}{}_{BA'} \qquad \mbox{and} \qquad
\sigma_{i}{}^{AB} \equiv -\mathrm{i}\rho^A{}_{A'}\sigma_{i}{}^{BA'},
\nonumber
\end{equation} with components
\begin{gather*} (\sigma_{\bm{0}}{}^{\bm{AB}}) =
(\sigma^{\bm{0}}{}_{\bm{AB}}) \equiv \frac{1}{\sqrt{2}}
    \begin{pmatrix} 0 & i \\ i & 0
    \end{pmatrix},\quad (\sigma_{\bm{1}}{}^{\bm{AB}}) =
(\sigma^{\bm{1}}{}_{\bm{AB}}) \equiv \frac{1}{\sqrt{2}}
    \begin{pmatrix} i & 0 \\ 0 & i
    \end{pmatrix},\\[8pt] (\sigma_{\bm{2}}{}^{\bm{AB}}) =
-(\sigma^{\bm{2}}{}_{\bm{AB}}) \equiv \frac{1}{\sqrt{2}}
    \begin{pmatrix} -1 & 0 \\ 0 & 1
    \end{pmatrix},\quad (\sigma_{\bm{3}}{}^{\bm{AB}}) =
(\sigma^{\bm{3}}{}_{\bm{AB}}) \equiv \frac{1}{\sqrt{2}}
    \begin{pmatrix} 0 & i \\ -i & 0
    \end{pmatrix}.
\end{gather*}

Again, as we are only interested in the symmetric parts, we see that
$\sigma_{\bm{3}}{}^{\bm{AB}} = \sigma^{\bm{3}}{}_{\bm{AB}}$ make no
contribution to the resulting $SU(1,1)$ spinor components and
$\sigma_{\bm{i}}{}^{\bm{AB}} = \sigma^{\bm{i}}{}_{\bm{AB}}$, $i=0,1,2$
are then the $SU(1,1)$ spinor Infeld-van der Waerden symbols. Thus,
the components of a vector can be translated into components of a
$SU(1,1)$ spinor using the correspondence
\begin{equation} (v^{\bm{a}}) =
(v^{\bm{0}}\;v^{\bm{1}}\;v^{\bm{2}}\;v^{\bm{3}}) \mapsto (v^{\bm{AB}})
= \frac{1}{\sqrt{2}}
    \begin{pmatrix} iv^{\bm{1}} - v^{\bm{2}} & iv^{\bm{0}} \\
iv^{\bm{0}} & iv^{\bm{1}} + v^{\bm{2}}
    \end{pmatrix}.
\end{equation} Again, the $SU(1,1)$ spinor of a real 4-vector is no
longer Hermitian. The reality of $v^{\bm{a}}$ now requires
$\hat{v}^{\bm{AB}} = v^{\bm{AB}}$.

\subsection*{2+1 decomposition of $SU(1,1)$ spinors} The crucial
difference between $SU(1,1)$ spinors and $SU(2)$ spinors is that the
former live on a 3-dimensional surface $\T$ with a metric that has
\emph{Lorentzian} signature. This means, we can perform a 2+1
splitting based on some time- like vector field $\tau^a$, intrinsic to
$\T$. To start, consider the spinors
\begin{equation} x_{AB} \equiv -\sqrt{2}o_{(A}\iota_{B)}, \qquad
y_{AB} \equiv -\frac{1}{\sqrt{2}}o_Ao_B, \qquad z_{AB} \equiv
\frac{1}{\sqrt{2}}\iota_A\iota_B.  \nonumber
\end{equation} It is clear that any symmetric valence-2 spinor can be
written as a linear combination of these three spinors. Using these,
we further define
\begin{eqnarray*} && \epsilon^0{}_{ABCD} \equiv 2z_{(AB}z_{CD)}, \\ &&
\epsilon^1{}_{ABCD} \equiv x_{(AB}z_{CD)}, \\ && \epsilon^2{}_{ABCD}
\equiv \frac12x_{(AB}x_{CD)}\equiv 2y_{(AB}z_{CD)}, \\ &&
\epsilon^3{}_{ABCD} \equiv -x_{(AB}y_{CD)}, \\ && \epsilon^4{}_{ABCD}
\equiv 2y_{(AB}y_{CD)},
\end{eqnarray*} As the above incorporate all five possible symmetrised
combinations of the spin-dyad, they form a basis for symmetric
valence-4 spinors.

\medskip We now consider the decomposition of totally-symmetric
unprimed spinors of valence 2 and 4, corresponding to $SU(1,1)$
spinors on $\T$. To implement the 2+1 decomposition, we introduce a
timelike vector field
\begin{equation} \tau^{AA'} = o^Ao^{A'} + \iota^A\iota^{A'}, \nonumber
\end{equation} intrinsic to $\T$ satisfying
$\tau^{AA'}\tau_{AA'}=2$. The corresponding $SU(1,1)$-spinor is then
\begin{equation} \tau_{AB} \equiv 2\mathrm{i}\,o_{(A}\iota_{B)} =
\sqrt{2}\mathrm{i}\,x_{AB}.  \nonumber
\end{equation} A valence-2 $SU(1,1)$ spinor $T_{AB}$ is a linear
combination of $x_{AB}, y_{AB}$ and $z_{AB}$. As these are all totally
symmetric they are intrinsic to $\T$. Noting that
\begin{equation} x_{AB}\tau^{AB} = -\mathrm{i}\sqrt{2},\qquad
y_{AB}\tau^{AB} = z_{AB}\tau^{AB} = 0, \nonumber
\end{equation} the decomposition of $T_{AB}$ with respect to
$\tau^{AB}$ is
\begin{equation} T_{AB} = (\alpha y_{AB} + \beta z_{AB}) +
\mu\tau_{AB} = \mu_{AB} + \mu \tau_{AB},\qquad
\alpha,\beta,\mu\in\mathbb{C}, \nonumber
\end{equation} where
\begin{equation} \mu_{AB} = \mu_{(AB)}, \qquad \mu_{AB}\tau^{AB} = 0.
\nonumber
\end{equation}

A valence-4 $SU(1,1)$ spinor $T_{ABCD}$ is a linear combination of
$\epsilon^i{}_{ABCD}$, $i=0,\ldots,4$. As these are all totally
symmetric they are intrinsic to $\T$. We have the relations
\begin{gather*} \epsilon^0{}_{ABCD}\tau^{AB} = 0, \\
\epsilon^1{}_{ABCD}\tau^{AB}\tau^{CD} = 0, \\
\epsilon^2{}_{ABCD}\tau^{AB}\tau^{CD} = -\frac23, \\
\epsilon^3{}_{ABCD}\tau^{AB}\tau^{CD} = 0, \\
\epsilon^4{}_{ABCD}\tau^{AB} = 0,
\end{gather*} which tell us which $\epsilon^i{}_{ABCD}$ have either
no, one or two $\tau^{AB}$ factors. It is clear $\epsilon^2{}_{ABCD}$
is the only one with two factors, and the irreducible decomposition
\begin{equation} \tau_{AB}\tau_{CD} =
-\frac23\epsilon_{A(C}\epsilon_{D)B} + \tau_{(AB}\tau_{CD)} \nonumber
\end{equation} then implies that
\begin{equation} \tau_{(AB}\tau_{CD)} = \tau_{AB}\tau_{CD} +
\frac23\epsilon_{A(C}\epsilon_{B)D} = \mu_2\epsilon^2{}_{ABCD},
\nonumber
\end{equation} for some $\mu_2\in\mathbb{C}$. Moving on to the
expressions for $\epsilon^1{}_{ABCD}$ and $\epsilon^3{}_{ABCD}$, each
of which must have a single factor of $\tau_{AB}$, we then find that
\begin{equation} \tau_{(AB}\mu_{CD)} = \frac12\Big{(}\tau_{AB}\mu_{CD}
+ \mu_{AB}\tau_{CD}\Big{)} = \mu_1\epsilon^1{}_{ABCD} +
\mu_3\epsilon^3{}_{ABCD}, \nonumber
\end{equation} for some $\mu_1,\mu_3\in\mathbb{C}$ and $\mu_{AB} =
\mu_{(AB)}$ satisfying $\mu_{AB}\tau^{AB} = 0$. Finally,
$\epsilon^0{}_{ABCD}$ and $\epsilon^4{}_{ABCD}$ must not contain any
$\tau_{AB}$ factors, and so we must have
\begin{equation} \mu_{ABCD} = \mu_0\epsilon^0{}_{ABCD} +
\mu_4\epsilon^4{}_{ABCD} \nonumber
\end{equation} for some $\mu_0,\mu_4\in\mathbb{C}$ and $\mu_{ABCD} =
\mu_{(ABCD)}$ satisfying $\mu_{ABCD}\tau^{AB} = 0$. Thus, a
totally-symmetric valence-4 spinor $T_{ABCD}$ can be decomposed with
respect to $\tau^a$ as
\begin{equation} T_{ABCD} = \mu_{ABCD} + \tau_{AB}\mu_{CD} +
\mu_{AB}\tau_{CD} - \mu\big{(} 3\tau_{AB}\tau_{CD} +
2\epsilon_{C(A}\epsilon_{B)D} \big{)},
    \label{DecompositionValence4}
\end{equation} where
\begin{equation} \mu_{ABCD} = \mu_{(ABCD)}, \quad \mu_{AB} =
\mu_{(AB)}, \quad \mu_{ABCD}\tau^{AB} = 0, \quad \mu_{AB}\tau^{AB} =
0.  \nonumber
\end{equation}



\end{document}